\documentclass{elsarticle}
\newcommand{\trans}{^{\mathrm{T}}}
\newcommand{\fring}{\mathcal{R}} %finite ring
\newcommand{\tlcycliccode}{$\theta$-$\lambda$-cyclic code}

\renewcommand{\leq}{\leqslant}
\newcommand{\Aut}{\mathrm{Aut\,}}
\newcommand{\id}{\mathrm{id}}
\newcommand{\wt}{\mathrm{wt}}
\newcommand{\ord}{\mathrm{ord\,}}
\newcommand{\cchar}{\mathrm{char\,}}
\newcommand{\dhamming}{\mathrm{d}_H}

\usepackage{amsmath}
\usepackage{amsfonts}
\usepackage{amssymb}
\usepackage{mathtools}
\usepackage{hyperref}
\usepackage{geometry}
\journal{arXiv}

\usepackage[amsmath,thmmarks]{ntheorem}
{
    \theoremstyle{nonumberplain}
    \theoremheaderfont{\bfseries}
    \theorembodyfont{\normalfont}
    \theoremsymbol{\mbox{$\Box$}}
    \newtheorem{proof}{Proof.}
}
\qedsymbol={\mbox{$\Box$}}
\newtheorem{theorem}{Theorem}[section]

\newtheorem{definition}[theorem]{Definition}
\newtheorem{proposition}[theorem]{Proposition}

\newtheorem{corollary}[theorem]{Corollary}
\newtheorem{lemma}[theorem]{Lemma}
\newtheorem{example}{Example}[section]
\newtheorem{remark}[example]{Remark}

\usepackage{verbatim}

\begin{document}

\begin{frontmatter}

\title{Skew constacyclic codes over a class of finite commutative semisimple rings}

\author{Ying Zhao}
\address{27 Shanda Nanlu, Jinan, P.R.China 250100}




\begin{abstract}
    Let $F_q$ be a finite field of $q=p^r$ elements, where $p$ is a prime, $r$ is a positive integer, we determine automorphism $\theta$ of a class of finite commutative semisimple ring $\fring =\prod_{i=1}^t F_q$ and the structure of its automorphism group $\Aut(\fring)$. We find that $\theta$ is totally determined by its action on the set of primitive idempotent $e_1, e_2,\dots,e_t$ of $\fring$  and its action on $F_q1_{\fring}=\{a1_{\fring} \colon a\in F_q\},$ where $1_{\fring}$ is the multiplicative identity of $\fring.$ We show $\Aut(\fring) = G_1G_2,$ where $G_1$ is a normal subgroup of $\Aut(\fring)$ isomorphic to the direct product of $t$ cyclic groups of order $r,$ and $G_2$ is a subgroup of $\Aut(\fring)$ isomorphic to the symmetric group $S_t$ of $t$ elements.

    For any linear code $C$ over $\fring,$  we establish a one-to-one correspondence between $C$ and $t$ linear codes $C_1,C_2,\dots,C_t$ over $F_q$ by defining an isomorphism $\varphi.$ For any $\theta$ in $\Aut(\fring)$ and any invertible element $\lambda$ in $\fring,$ we give a necessary and sufficient condition that a linear code over $\fring$ is a $\theta$-$\lambda$-cyclic code in a unified way. When $ \theta\in G_1,$ the $C_1,C_2,\dots,C_t$ corresponding to the $\theta$-$\lambda$-cyclic code $C$ over $\fring$ are skew  constacyclic codes over $F_q.$ When $\theta\in G_2,$ the $C_1,C_2,\dots,C_t$ corresponding to the skew cyclic code $C$ over $\fring$ are quasi cyclic codes over $F_q.$ For general case, we give conditions that $C_1,C_2,\dots,C_t$ should satisfy when the corresponding linear code $C$ over $\fring$ is a skew constacyclic code.

    Linear codes over $\fring$ are closely related to linear codes over $F_q.$ We define homomorphisms which map linear codes over $\fring$ to matrix product codes over $F_q.$ One of the homomorphisms is a generalization of the $\varphi$ used to decompose linear code over $\fring$ into linear codes over $F_q,$ another homomorphism is surjective. Both of them can be written in the form of $\eta_M$ defined by us, but the matrix $M$ used is different. As an application of the theory constructed above, we construct some optimal linear codes over $F_q.$
\end{abstract}


\end{frontmatter}

%\linenumbers

\section{Introduction}

In modern society, the efficient and reliable transmission of information is inseparable from coding theory. Cyclic codes have been widely studied and applied due to their excellent properties. In recent years, more and more people have studied cyclic codes. Early studies mainly considered codes over finite fields, but when Hammons et al. \cite{Z4linear} found that some good nonlinear codes over binary fields can be regarded as Gray images of linear codes over $\mathbb{Z}_4,$ more and more people began to study linear codes over finite rings, and by defining appropriate mappings, codes over finite rings were related to codes over finite fields.

 Boucher et al. studied skew cyclic codes over finite fields \cite{Boucher2007Skew, 2009Coding}, and gave some self dual linear codes over $F_4$ with better parameters than known results, which attracted many people to study skew cyclic codes. Skew cyclic codes is a generalization of the cyclic codes, and it can be characterized by skew polynomial ring $F_q\left[x;\theta\right].$ The elements in skew polynomial ring are still polynomials, but multiplication is non-commutative, and factorization is not unique. It is because of these differences that Boucher et al. obtained self dual linear codes with batter parameters over finite fields. But they required $\ord (\theta) $ divides the length of codes. Then, Siap et al. \cite{siap2011skew} removed this condition and studied skew cyclic codes of arbitrary length.

Boucher et al. also studied skew constacyclic codes over $F_q$ \cite{Boucher2009codesMod, boucher2011note} and skew constacyclic codes over Galois rings \cite{boucher2008skew}. Influenced by the work of Boucher et al., people began to research skew cyclic codes and skew constacyclic codes over finite rings. But general finite rings no longer have many good properties like finite fields, so it is not easy to study linear codes over general finite rings. Therefore, people always consider a very specific ring or a class of rings, and then study linear codes over the ring. Abualrub et al. \cite{2012Onskewcyclic} studied skew cyclic codes and its dual codes over the ring $F_2+vF_2, v^2=v,$ they used skew polynomial rings.

Dougherty et al. \cite{dougherty1999self} used the Chinese remainder theorem to study self dual codes over $\mathbb{Z}_{2k},$ which brought new ideas for the study of linear codes over finite commutative rings. Similar to the method used by Zhu et al. \cite{zhu2010some} to study cyclic codes over $F_2+vF_2,$ when people study linear codes over $F_q+vF_q, v^2=v,$ they mainly use $F_q +vF_q$ is isomorphic to the finite commutative semisimple ring $F_q\times F_q,$ and then they show that to study a linear code $C$ over the ring, it is only necessary to study the two linear codes $C_1,C_2$ over $F_q.$ Gao \cite{gao2013skew} studied skew cyclic codes over $F_p +vF_p,$ and the main result he obtained can be regarded as a special case of the result in special case two we will deal with later. Gursoy et al. \cite{gursoy2014construction} studied skew cyclic codes over $F_q +vF_q,$ their main result can be seen as a special case of the result in special case one we deal with later. Gao et al. \cite{gao2017skew} later studied skew constacyclic codes over $F_q+vF_q,v^2=v,$ what they considered can be seen as special case one in this paper.

A more general ring than the one above is $F_q\left[v\right]/\left(v^{k+1}-v\right),\ k\mid (q-1),$ which is isomorphic to finite commutative semisimple ring $\prod_{j=1}^{k+1} F_q.$ Shi et al. \cite{shi2015skew} studied skew cyclic codes over the ring $F_q+vF_q+v^2F_q,$ $v^3=v,$ $q=p^m, $ where $p$ is an odd prime number, they also used the structure of the ring, and then turned the problem into the study of linear codes over $F_q,$ the skew cyclic codes they studied can be included in the special case one of this paper.

The ring $F_q[u,v]/(u^2-u,v^2-v)$ is isomorphic to the finite commutative semisimple ring $F_q\times F_q\times F_q\times F_q,$ the study of linear codes over this ring is similar, Yao et al. \cite{ting2015skew} studied skew cyclic codes over this ring, although they chose a different automorphism of this ring, they limited linear codes $C$ over the ring satisfy specific conditions. Islam et al. \cite{islam2018skew} also studied skew cyclic codes and skew constacyclic codes over this ring, the automorphism they chose can be classified into our special case one.

Islam et al. \cite{islam2019note} also studied skew constacyclic codes over the ring $F_q+uF_q+vF_q$ ($u^2=u,v^2=v,uv=vu=0$), they mainly used the ring is isomorphic to the finite commutative semisimple ring $F_q\times F_q\times F_q,$ the skew cyclic codes they studied can also be classified as special case one in this paper. Bag et al. \cite{bag2019skew} studied skew constacyclic codes over $F_p+u_1F_p + \dots + u_{2m}F_p,$ this ring is isomorphic to $\prod_{j=1}^{2m+1} F_p,$ their results are similar to those in our special case two. The main work they do is to define two homomorphisms, so that the skew constacyclic codes over the ring will be mapped to special linear codes over $F_p.$
% Cyclic code or quasi-cyclic code.

There are also many people who study linear codes over rings with similar structures and then construct quantum codes. For example, Ashraf et al. \cite{Ashraf2019Quantum} studied cyclic codes over $F_p\left[u,v\right]/\left(u^2 -1,v^3-v\right),$ and then using the obtained result to construct quantum codes, this ring is isomorphic to $\prod_{j=1}^{6} F_p.$ Bag et al. \cite{bag2020quantum} studied skew constacyclic codes over $F_q [u, v]/\langle u^2- 1, v^2- 1, uv- vu\rangle$ to construct quantum codes, this ring is isomorphic to $F_q\times F_q \times F_q \times F_q,$ the part of their paper on skew constacyclic codes can be classified into our special case one.

It is noted that the above rings are all special cases of a class of finite commutative semisimple rings $\fring=\prod_{i=1}^tF_q,$ and the methods used in the related research and the results obtained are similar, so we are inspired to study skew constacyclic codes over $\fring.$ In addition, we noticed that when the predecessors studied skew constacyclic codes over rings, they often considered not all but a given automorphism, the results are not very complete, so we intend to characterize skew constacyclic codes over $\fring$ corresponding to all automorphisms. However, in order to make the results easier to understand, we also discuss two special cases.

 Dinh et al. have studied constacyclic codes and skew constacyclic codes over finite commutative semisimple rings\cite{dinh2017constacyclic,Dinh2019Skew}, but what they discussed were similar to our special case one. Our results can be further generalized to skew constacyclic codes over general finite commutative semisimple rings. But to clarify the automorphism of general finite commutative semisimple rings, the symbol will be a bit complicated, and when we want to establish relations from linear codes over rings to codes over fields, what we use is essentially this kind of special finite commutative semisimple rings we considered, so we do not consider skew constacyclic codes over general finite commutative semisimple rings here.

In the rest of this paper, we first introduce linear codes over $F_q,$ and review the characterization of skew constacyclic codes and their dual codes over $F_q.$ In order to clearly describe skew constacyclic codes over $\fring,$ We determine the automorphism of $\fring$ and the structure of its automorphism group, we define an isomorphism to decompose the linear code over $\fring,$ and give the characterizations of skew constacyclic codes in two special cases, and finally the characterization of the general skew constacyclic codes over $\fring$ is given. For the last part of this paper, we define homomorphisms to relate linear codes over $\fring$ to matrix product codes over $F_q,$ and give some optimal linear codes over $F_q.$

\section{Skew constacyclic codes over \texorpdfstring{$F_q$}{Fq}}

Let $F_q$ be a $q$ elements finite field, and if $C$ is a nonempty subset of the $n$ dimensional vector space $F_q^n$ over $F_q,$ then $C$ is a code of length $n$ over $F_q,$ and here we write the elements in $F_q^n$ in the form of row vectors. The element in a code $C$ is called codeword, and for $x\in C,$ denote by $\wt_H(x)$ the number of non-zero components in $x,$ which is called the Hamming weight of $x.$ For two elements $x,y\in C$ define the distance between them as $\dhamming(x,y) = \wt_H(x-y). $ If $C$ has at least two elements, define the minimum distance $\dhamming(C) = \min\{\dhamming(x,y) \colon x, y\in C,\, x\neq y\}$ of $C$.  Define the inner product of any two elements $x=(x_0,x_1,\dots,x_{n-1}), y=(y_0,y_1,\dots,y_{n-1}) \in F^n_q$ as $x\cdot y = \sum_{i=0}^{n-1} x_iy_i.$ With the inner product we can define the dual of the code $C$ as $C^\perp = \{x\in F_q^n \colon x\cdot y=0, \, \forall\, y \in C\},$ the dual code must be linear code we define below.

\begin{definition}
    If $C$ is a vector subspace of $F_q^n,$ then $C$ is a linear code of length $n$ over $F_q.$
\end{definition}

If $C$ is a nonzero linear code, then there exists a basis of $C$ as a vector subspace, for any basis $\alpha_1,\alpha_2,\dots,\alpha_k,$ where $k$ is the dimension of $C$ as a linear space over $F_q,$ then we say that the $k\times n$ matrix $\left(\alpha_1,\alpha_2,\dots,\alpha_k\right)\trans$ is the generator matrix of $C.$ Suppose $H$ is a generator matrix of $C^\perp,$ then $x\in C$ if and only if $Hx\trans = 0,$ that is, we can use $H$ to check whether the element $x$ in $F_q^n$ is in $C,$ and call $H$ a parity-check matrix of $C.$ Let $C$ be a linear code, if $C \subset C^\perp,$ then $C$ is a self orthogonal code; if $C = C^\perp ,$ then $C$ is a self dual code.

% For a linear code with code length $n,$ dimension $k,$ minimum distance $d$ on $F_q$, code length $n,$ dimension $k,$ minimum distance $d$ is the most important parameter of $C$ of the linear code, which is generally denoted by $[n,k,d]$ for $C$, and $[n,k]$ for a linear code with code length $n,$ dimension $k.$ if the minimum distance $d,$ is not specified. $k.$ For a linear code $C,$ with $[n,k,d]$ parameters on $F_q$, its dimension $k$ and minimum distance $d$ are mutually constrained, the well-known Singleton bound is that $k+d \leq n+1,$ which gives an upper bound on the minimum distance $d$ after a fixed code length $n$ and dimension $k$. There are other bounds. A linear code $C$ is said to be optimal if it reaches a theoretical lower bound on the minimum distance.

If $C$ is a linear code with length $n,$ dimension $k,$ and minimum distance $d,$ then we say it has parameters $[n,k,d].$ Sometimes, without indicating the minimum distance of $C,$ we say the parameters of $C$ is $[n,k]. $ For given positive integer $n,k,$ where $1\leq k \leq n,$ define the maximum value of the minimum distance of all linear codes over $F_q$ with parameters $[n,k]$ as $\dhamming(n,k,q) = \max\left\{ \dhamming(C) \colon C\subset F_q^n,\, \dim C=k \right\}. $ In general, it is difficult to determine the exact value of $\dhamming(n,k,q)$, often only the upper and lower bounds are given, the famous Singleton bound is that $\dhamming(n,k,q) \leq n-k+1.$
A linear code $C$ with parameter $[n,k]$ is said to be an optimal linear code over $F_q$ if the minimum distance $\dhamming(C)$ is exactly equal to $\dhamming(n,k,q).$

% Let $C_1$ and $C_2$ be linear codes of code length $n$ on $F_q$, and $C_2$ is said to be equivalent to $C_1$ if the elements in $C_2$ are obtained from the elements in $C_1$ by column permutation and multiplication of each column by a nonzero element in $F_q$. Obviously, the two equivalent linear codes have the same parameters. 

The most studied linear codes are cyclic codes, and the definitions of cyclic codes and quasi cyclic codes over $F_q$ are given below.

\begin{definition}
    Let $\rho$ be a map from $F_q^n$ to $F_q^n$, which maps $ (c_0,c_1,\dots,c_{n-1})$ to $\left(c_{n-1},c_0,\dots,c_{n-2}\right). $ Let $C$ be a linear code of length $n$ over $F_q$ and $\ell$ be a positive integer, if $\rho^{\ell}(C) = C,$ then $C$ is a quasi cyclic code of length $n$ with index $\ell$ over $F_q.$ In particular, $C$ is said to be a cyclic code when $\ell =1.$
\end{definition}

From the definition, it is clear that quasi cyclic code is a generalization of cyclic code, and another generalization of cyclic code is skew constacyclic code.

\subsection{Characterization of skew constacyclic codes over \texorpdfstring{$F_q$}{Fq}}\label{sec:skewF_q}

\begin{definition}
    Let $\theta\in \Aut(F_q),$ $\lambda\in F_q^\ast,$ and use $\rho_{\theta,\lambda}$ to denote the mapping from $F_q^n$ to $F_q^n$, which maps $ (c_0,c_1,\dots,c_{n-1})$ to $\left(\lambda \theta(c_{n-1}),\theta(c_0),\dots,\theta(c_{n-2})\right). $ If $C$ is a linear code of length $n$ over $F_q,$ and for any $c\in C,$ with $\rho_{\theta,\lambda}(c)\in C,$ then $C$ is a $\theta$-$\lambda$-cyclic code of length $n$ over $F_q$, and if $\theta$ and $\lambda$ are not emphasized, $C$ is also called a skew constacyclic code. In particular, if $\theta = \id,$  $C$ is a cyclic code when $\lambda=1$, a negative cyclic code when $\lambda=-1,$ and a $\lambda$-cyclic code or constacyclic code when $\lambda$ is some other invertible element. If $\lambda=1,\theta\neq \id ,$ then it is called a $\theta$-cyclic code, also called a skew cyclic code.
\end{definition}

Skew polynomials were first studied by Ore \cite{Ore1933poly}, and a detailed description of skew polynomial ring $F_q[x;\theta]$ over a finite field $F_q$ can be found in McDonald's monograph \citep[\uppercase\expandafter{\romannumeral2}. (C),][]{McDonald1974FiniteRW}.

\begin{definition}
    Let $\theta\in \Aut(F_q),$ define the ring of skew polynomials over the finite field $F_q$
    \begin{equation*}
        F_q[x;\theta] = \left\{a_0 +a_1x+\dots+a_kx^k\mid a_i \in F_q, 0\leq i \leq k\right\}.
    \end{equation*}
    That is, the element in $F_q[x;\theta]$ is element in $F_q[x]$, except that the coefficients are written on the left, which is called a skew polynomial, and the degree of skew polynomial is defined by the degree of polynomial. Addition is defined in the usual way, while $ax^n \cdot (bx^m) = a\theta^n(b)x^{n+m}, $ and then the general multiplication is defined by the law of associativity and the law of distribution.
\end{definition}

If $\theta = \id ,$ then $F_q[x;\theta]$ is $F_q[x].$ If $\theta \neq \id,$ then $F_q[x;\theta]$ is a non-commutative ring, with properties different from $F_q[x]$, such as right division algorithm.

\begin{theorem}%[right with remainder division, \citep{McDonald1974FiniteRW}]
    \citep[Theorem \uppercase\expandafter{\romannumeral2}.11,][]{McDonald1974FiniteRW}
    For any $ f(x) \in F_q[x;\theta], 0\neq g(x) \in F_q[x;\theta],$ there exists unique $q(x),r(x) \in F_q[x;\theta]$ such that
    $f(x) = q(x)g(x) + r(x),$ where $r(x) = 0$ or $0\leq \deg r(x) < \deg g(x). $
\end{theorem}

\begin{proof}
    Similar to the proof in the polynomial ring, which is obtained by induction.
\end{proof}

For convenience, use $\left\langle x^n-\lambda\right\rangle$ to abbreviate the left $F_q[x;\theta]$-module $F_q[x;\theta]\left(x^n-\lambda\right)$ generated by $x^n-\lambda$, then left $F_q[x;\theta]$-module $R_n = F_q[x;\theta]/\left\langle x^n-\lambda\right\rangle$ is the set of equivalence classes obtained by dividing the elements in $F_q[x;\theta]$ by $x^n-\lambda$ using right division algorithm. Define a map $\Phi:$
\begin{align*}
    F_q^{n} & \rightarrow F_q[x;\theta]/\left\langle x^n-\lambda\right\rangle \\
    (c_0,c_1,\dots,c_{n-1}) & \mapsto c_0+c_1x+\dots+c_{n-1}x^{n-1}+\left\langle x^n-\lambda\right\rangle.
\end{align*}

Clearly $\Phi$ is an isomorphism of vector spaces over $F_q.$

The following theorem is a general form of \citep[Theorem 10,][]{siap2011skew} by Siap et al.. They consider the case $\lambda =1$. It is worth noting that Boucher et al. should also have noticed the fact that they used submodules to define the module $\theta$-constacyclic codes \citep[Definition 3,][]{Boucher2009codesMod}.

\begin{theorem} \label{thm:skewcodesoverFbymod}
    Let $\theta \in \Aut(F_q),$  $\lambda \in F_q^\ast,$  $C$ be a vector subspace of $F_q^{n},$ 
    then $C$ is a \tlcycliccode\ if and only if $ \Phi(C)$ is a left $F_q[x;\theta]$-submodule
    of $R_n.$
\end{theorem}
\begin{proof}
    Necessity. Note that
    \begin{align*}
         & \mathrel{\phantom{=}} x\cdot\left(c_0+c_1x+\dots+c_{n-1}x^{n-1}+\left\langle x^n-\lambda\right\rangle\right) \\
         & = \theta(c_0) x + \theta(c_1)x^2+\dots+\theta(c_{n-1})x^n + \left\langle x^n-\lambda\right\rangle \\
         & = \lambda\theta(c_{n-1}) + \theta(c_0)x+\dots+\theta(c_{n-2})x^{n-1} + \left\langle x^n-\lambda\right\rangle \\
         & = \Phi\left(\lambda\theta(c_{n-1}),\theta(c_0),\dots,\theta(c_{n-2})\right) \in \Phi(C),
    \end{align*}
    Thus for any $a(x) \in F_q[x;\theta],c(x) + \left\langle x^n-\lambda\right\rangle \in \Phi(C),$ we have $a(x)\left(c(x) + \left\langle x^n-\lambda\right\rangle \right) \in \Phi(C).$

    Sufficiency. For any $(c_0,\dots,c_{n-1}) \in C,$ we need to prove that $$\left(\lambda\theta(c_{n-1}),\theta(c_0),\dots,\theta(c_{n-2})\right) \in C.$$ Notice that $c_0+c_1x+\dots+c_{n-1}x^{n-1}+\left\langle x^n-\lambda\right\rangle \in \Phi(C),$ and $ \Phi(C)$ is a left $F_q[x;\theta]$-submodule of $R_n,$ so
    $x \cdot \left(c_0+c_1x+\dots+c_{n-1}x^{n-1}+\left\langle x^n-\lambda\right\rangle\right) \in \Phi(C),$ which gives the proof.
\end{proof}

Each nonzero element in the left $F_q[x;\theta]$-module $R_n$ corresponds to a skew polynomial in $F_q\left[x;\theta \right]$ with degree no more than $n-1$. If $C$ is a \tlcycliccode\ of length $n$ over $F_q$ and $C\neq \{0\},$ then there exists a monic skew polynomial $g(x)$ with minimal degree such that $g(x) + \left\langle x^n-\lambda\right\rangle \in \Phi(C),$ then $F_q[x;\theta]\left(g(x)+\left\langle x^n-\lambda\right\rangle\right) = \Phi(C).$ This is because for each $ f(x) + \left\langle x^n-\lambda\right\rangle \in \Phi(C),$
by right division algorithm we have $f(x) = q(x)g(x) + r(x),$ where $r(x) = 0$ or $0\leq \deg r(x) < \deg g(x). $
The former case corresponds to $f(x) + \left\langle x^n-\lambda\right\rangle = q(x)\left(g(x)+\left\langle x^n-\lambda\right\rangle\right). $ In the latter case $r(x)+\left\langle x^n-\lambda\right\rangle = f(x)+\left\langle x^n-\lambda\right\rangle - q(x)\left(g(x)+\left\langle x^n-\lambda\right\rangle\right) \in \Phi(C),$ so that a monic skew polynomial of lower degree can be found, it is a contradiction. Similarly, it can be shown that the monic skew polynomial $g(x)$ with the minimal degree is unique.

\begin{definition}
    Let $C$ be a \tlcycliccode\ of length $n$ over $F_q$ and $C\neq \{0\},$ define $g(x)$ as described above to be the  generator skew polynomial of $C,$ we also call it generator polynomial when $\theta = \id.$
\end{definition}

The generator skew polynomial $g(x)$ should be a right factor of $x^n-\lambda.$ In fact, according to the right division algorithm $x^n-\lambda = q(x)g(x) + r(x),$ if $r(x)\neq 0,$ then $r(x)+\left\langle x^n-\lambda\right\rangle = -q(x)\left(g(x) + \left\langle x^n-\lambda\right\rangle\right) \in \Phi(C),$
contradicts with  $g(x)$ is the monic skew polynomial of minimal degree. Let $g(x) = a_0+a_1x+\dots + a_{n-k}x^{n-k},$ then one of the generator matrices of $C$ is

\begin{equation}\label{eq:genmat}
    G = \begin{pmatrix}
        a_0 & \dots & a_{n-k} & & & \\
            & \theta(a_0) & \dots & \theta(a_{n-k}) & & \\
            & & \ddots & \ddots & \ddots & \\
            & & & \theta^{k-1}(a_0) & \dots & \theta^{k-1}(a_{n-k})
    \end{pmatrix}.
\end{equation}

If $C$ is a nontrivial (i.e., $C\neq 0$ and $C\neq F_q^n$) \tlcycliccode\ of length $n,$ there is a monic skew polynomial $g(x) = a_0+a_1x+\dots + a_{n-k}x^{n-k}, (a_{n-k}= 1),$ and $g(x)$ right divides $x^n-\lambda.$ Conversely, for a monic right factor $g(x)$ with degree less than $n$ of $x^n-\lambda,$ the left $F_q[x;\theta]$-module generated by $ g(x) + \langle x^n-\lambda \rangle$ corresponds to a \tlcycliccode\ of length $n$ over $F_q$. Thus, we establish a one-to-one correspondence between the nontrivial \tlcycliccode\ of length $n$ over $F_q$ and the nontrivial monic right factor of $x^n-\lambda$ in $F_q[x;\theta].$ Unfortunately, however, since factorization in skew polynomial ring is different from factorization in polynomial ring, we cannot use this correspondence to explicitly give a formula counting the number of \tlcycliccode\ over $F_q.$ We can concretely perceive the difference in factorization from the following example.

\begin{example}
    Consider $F_4=\{0,1,\alpha, 1+\alpha\},$ where $\alpha^2 = 1+\alpha.$ Let $\theta$ be a nontrivial automorphism of $F_4.$ It is straightforward to verify that $x^3-\alpha$ is irreducible in $F_4[x]$, while it can be decomposed in $F_4[x;\theta]$ as $x^3-\alpha = (x+\alpha)(x^2+\alpha^2 x + 1) = (x^2+\alpha^2 x + 1)(x+\alpha). $ Also, $x^3-1$ can be decomposed in $F_4[x]$ as $x^3-1 = (x-\alpha)(x-\alpha^2)(x-1),$ but in $F_4[x;\theta]$ it cannot be decomposed as a product of linear factors, only as $x^3-1 = (x^2 + x+1)(x-1) = (x-1)(x^2+x+1).$
\end{example}

\begin{definition}
    Let $\lambda \in F_q^\ast,$ and use $\rho_{\lambda} $ to denote the mapping from $F_q^n$ to $F_q^n$, which maps $ (c_0,c_1,\dots,c_{n-1})$ to $\left(\lambda c_{n-1},c_0,\dots,c_{n-2} \right). $ Let $C$ be a linear code of length $n$ over $F_q$ and $\ell$ be a positive integer, if $\rho_{\lambda}^{\ell}(C) = C,$ then $C$ is a quasi $\lambda$-cyclic code of length $n$ with index $\ell$ over $F_q$.
\end{definition}

Siap et al. studied skew cyclic codes of arbitrary length over finite field, their main results \citep[Theorem 16 and Theorem 18,][]{siap2011skew} are that skew cyclic codes over $F_q$ are either cyclic codes or quasi cyclic codes. Their results can be generalized to the following

\begin{theorem}
    Let $\theta\in \Aut(F_q), \lambda \in F_q^\ast ,$ $\ord (\theta) = m, $ $\theta(\lambda) = \lambda,$  $C$ is a $\theta$-$\lambda$-cyclic code of length $n$ over $F_q.$ If $(m,n)=1,$ then $C$ is a $\theta$-$\lambda$-cyclic code of length $n$ over $F_q$. If $(m,n)=\ell,$ then $C$ is a quasi $\lambda$-cyclic code over $F_q$ of length $n$ with index $\ell.$
\end{theorem}

\begin{proof}
    We only need to proof the case $(m,n)=\ell.$ In this case, define a mapping $\tilde{\theta}$ from $F_q^n$ to $F_q^n$ which maps $(x_0,x_1,\dots,x_{n-1})$ to $(\theta(x_0),\theta(x_1),\dots,\theta(x_{n-1})). $
    You can directly verify $\rho_{\lambda} \circ \tilde{\theta} = \tilde{\theta} \circ \rho_{\lambda} = \rho_{\theta,\lambda}. $ Since $(m,n) = \ell, $ so there exists $a,b \in \mathbb{N},$ such that $am = \ell + bn.$

    For any $x = (x_0,x_1,\dots,x_{n-1}) \in C,$ since $C$ is a \tlcycliccode, therefore
    \begin{equation*}
        \rho_{\theta,\lambda}^{am} (x) = \rho_{\lambda}^{\ell + bn} \tilde{\theta}^{am} (x) = \rho_{\lambda}^{\ell + bn}(x) = \lambda^b \rho_{\lambda}^{ \ell} (x) \in C,
    \end{equation*}
    so $\rho_{\lambda}^{ \ell} (x) \in C.$
\end{proof}

\begin{remark}
    The converse of the above theorem does not hold, i.e., quasi cyclic code over $F_q$ is not necessarily a skew cyclic code. For example, let $F_4=\{0,1,\alpha, 1+\alpha\},$ where $\alpha^2 = 1+\alpha.$ Let $\theta$ be the nontrivial automorphism of $F_4.$ Consider the code $C=\{(0,0,0),(\alpha,\alpha^2,1),(\alpha^2,1,\alpha),(1,\alpha,\alpha^2)\}$ of length $3$ over $F_4,$ and it is straightforward to verify that $C$ is a cyclic code over $F_4,$ but it is not a $\theta$-cyclic code over $F_4.$ 
\end{remark}

\subsection{Dual codes of skew constacyclic codes over \texorpdfstring{$F_q$}{Fq}}

The following theorem is well known, and we give a simple proof similar to the proof of \citep[Theorem 2.4,][]{valdebenito2018dual} by Valdebenito et al.

\begin{theorem} \label{dualcodesoverF}
    If $C$ is a \tlcycliccode\ of length $n$ over $F_q$, where $\theta\in \Aut(F_q),\lambda \in F_q^\ast, $ then $C^\perp$ is a $\theta$-$\lambda^{-1}$-cyclic code of length $n$ over $F_q.$
\end{theorem}
\begin{proof}
    For any $z=(z_0,z_1,\dots,z_{n-1}) \in C^\perp,$ we need to prove that $$\tilde{z}=\left(\lambda^{-1}\theta(z_{n-1}),\theta(z_0),\dots,\theta(z_{n-2})\right) \in C^\perp.$$ Since $C$ is a \tlcycliccode, for any $y\in C,$ it can be written as $$y= \left(\lambda\theta(c_{n-1}),\theta(c_0),\dots,\theta(c_{n-2})\right),$$ where $(c_0,c_1,\dots,c_{n-1}) \in C.$ Thus $ \tilde{z} \cdot y= 0,$ and by the arbitrariness of $y$ we get $\tilde{z}\in C^\perp.$
\end{proof}

If $C$ is a nontrivial \tlcycliccode\ of length $n$ over $F_q$, then $C^\perp$ is a $\theta$-$\lambda^{-1}$-cyclic code of length $n$ over $F_q$ by Theorem \ref{dualcodesoverF}. Naturally, $C^\perp$ also has generator skew polynomial $g^\perp(x),$ and since $C$ uniquely determines $C^\perp,$ it follows that $g^\perp(x)$ is uniquely determined by the generator skew polynomial $g(x)= a_0+a_1x+\dots + a_{n-k}x^{n-k}$ of $C.$ In fact, if $g^\perp(x)$ corresponds to the element $c=(c_0,c_1,\dots,c_{n-1}),$ then since $g^\perp(x)$ is degree $k$ and monic we get: $c_k = 1, c_i=0, k<i<n,$ only $c_0,c_1,\dots ,c_{k-1}$ is still need to be determined, and then the system of linear equations $Gc\trans =0$ is uniquely solved for $c_i, 0\leq i\leq k-1,$ where $G$ is given by the Eq. \eqref{eq:genmat}. Thus, the coefficients of $g^\perp(x)$ can be expressed by the coefficients of $g(x)$, but the formula is a bit complicated and is omitted here.

Although it is complicated to give a specific expression for $g^\perp(x)$ in terms of the coefficients of $g(x)$ directly, there are simple formulas to write down the relationship between $g(x)$ and $g^\perp(x)$ indirectly. Let us first prove a technical lemma, originally given by \citep[Lemma 2,][]{boucher2011note} of Boucher et al.

\begin{lemma} \label{lem:dualpoly}
    Let $\lambda \in F_q^\ast,$  $\theta \in \Aut(F_q),$ if $g(x)$ is monic with degree $n-k,$ and $x^n- \lambda = h(x)g(x)$ in $F_q[x;\theta],$ then $x^n - \theta^{ -k}(\lambda) = g(x)\lambda^{-1}h(x)\theta^{-k}(\lambda). $
\end{lemma}

\begin{proof}
    Let $g(x) = a_0 + a_1 x + \dots + a_{n-k}x^{n-k},$ denote $g_\theta(x) = \theta^{n}(a_0) + \theta^{n}(a_1) x + \dots + \theta^{n}(a_{n-k}) x^{n-k},$ then $x^n g(x) = g _\theta(x) x^n,$
    and thus
    \begin{align*}
        \left(x^n - g_\theta(x) h(x) \right) g(x) & = x^n g(x) - g_\theta(x) h(x)g(x) \\\
                                                    & = g_\theta(x) \left(x^n - h(x)g(x) \right) \\\
                                                    & = g_\theta(x)\lambda.
    \end{align*}
    From the fact that both sides of the equation have the same degree and the same coefficient of the highest term we get $x^n - g_\theta(x) h(x) = \theta^{n-k}(\lambda)$ and $\theta^{n-k}(\lambda) g(x) = g_\theta(x)\lambda,$ so $x^n - \theta^{n-k}(\lambda)$ $\theta^{n-k}(\lambda) = g_\theta(x)h(x) = \theta^{n-k}(\lambda) g(x) \lambda^{-1} h(x), $ multiple the left side of the equation by $\theta^{n-k}(\lambda^{-1})$ the right side by $\theta^{-k} (\lambda)$ to get $x^n - \theta^{-k}(\lambda) = g(x) \lambda^{-1}h(x) \theta^{-k}(\lambda). $
\end{proof}

The following theorem is derived from \citep[Theorem 1,][]{boucher2011note} by Boucher et al. and can also be found in \citep[Proposition 1,][]{valdebenito2018dual} by Valdebenito et al.
\begin{theorem} \label{polynomialofdualcodes}
    Let $\theta \in \Aut(F_q),$  $\lambda \in F_q^\ast,$  $C$ be a \tlcycliccode\ of length $n$ over $F_q$ and $C$ is not $\{0\}$ and $F_q^{n},$  $g(x)= a_0+a_1x + \dots + a_{n-k}x^{ n-k}$ is the generator skew polynomial of $C,$ $x^n - \lambda=h(x) g(x),$ denote $\hbar(x) = \lambda^{-1}h(x) \theta^{-k}(\lambda) = b_0 + b_1x + \dots + b_kx^k,$  $\hbar^{\ast} (x) = b_ k + \theta(b_{k-1}) x + \dots + \theta^k(b_0) x^k,$
    then the generator skew polynomial of $C^\perp$ is $ g^\perp(x) = \theta^k(b_0^{-1}) \hbar^{\ast}(x). $
\end{theorem}

\begin{proof}
    For $i>n-k,$ let $a_i =0.$ For $j<0,$ let $b_j =0.$
    By Lemma \ref{lem:dualpoly} we know that $x^n - \theta^{-k}(\lambda) = g(x) \hbar(x),$
    comparing the coefficients of $x^k$ on both sides of the equation gives
    \begin{equation*}
        a_0b_k + a_1\theta(b_{k-1}) + \dots + a_k\theta^k (b_0) = 0.
    \end{equation*}
    Comparing the coefficients of $x^{k-1}$ yields
    \begin{equation*}
        a_0b_{k-1} + a_1\theta(b_{k-2})+ \dots + a_{k-1}\theta^{k-1}(b_0) = 0,
    \end{equation*} act with $\theta$ to get
    \begin{equation*}
        \theta(a_0)\theta(b_{k-1}) + \theta(a_1)\theta^2(b_{k-2}) + \dots + \theta(a_{k-1})\theta^{k}(b_0) = 0.
    \end{equation*}
    And so on, comparing the coefficients, and then using the $\theta$ action, we get $k$ equations. The last one is to compare the coefficients of $x$ to get $a_0b_1 + a_1 \theta(b_0) = 0,$ so $\theta^{k-1}(a_0) \theta^{k-1}(b_1) + \theta^{k-1}(a_1)\theta^{k}(b_0) = 0.$

    Observe that a generator matrix of $C,$ i.e., a parity-check matrix $G$ of $C^\perp$ yields $$\left(b_k,\theta(b_{k-1}),\theta^2(b_{k-2}),\dots,\theta^{k}(b_0),0,\dots,0\right)$$ belonging to $C^\perp,$ where $ G$ is specified in Eq. \eqref{eq:genmat}. Notice that $\dim C^\perp = n-k,$ so the degree of the generator skew polynomial of $C^\perp$ is $k,$ and the degree of $\hbar^{\ast}(x) = b_k + \theta(b_{k-1}) x + \dots + \theta^k(b_0) x^k$ is $k,$
    thus $\theta^k(b_0^{-1}) \hbar^{\ast}(x)$ is the generator skew polynomial of $C^\perp.$
\end{proof}
%\begin{remark}
% $ \hat{h} (x) = \theta^{-n}(b_s) + \theta^{-n+1}(b_{s-1}) x + \dots + \theta^{-n+s}(b_0)x^s,$ then $\tilde{g}(x)\hat{h} (x) = 1-\lambda x^n.$ If $\ord( \theta) \mid n, \theta(\lambda) = \lambda$ then $\hat{h}(x) = h^\ast(x)$ and $1-\lambda x^n \in Z\left(F_q[x;\theta]\right),$ so $\tilde{g}(x){h^\ast( x)} = 1-\lambda x^n= h^\ast(x)\tilde{g}(x). $
%\end{remark}

Assume the same as the above theorem, let $\tilde{g}(x) = \theta^{-k}(a_k) + \theta^{-(k-1)}(a_{k-1})x + \dots + a_0x^k,$ then it is straightforward to verify that $1-\lambda x^n = \hbar^\ast (x) \tilde{g}(x). $ One can use skew polynomials to determine whether a linear code has an inclusion relation with its dual.

\begin{theorem}\label{containingcondition}
    Let $\theta \in \Aut(F_q),$  $\lambda \in F_q^\ast,$ if $C$ is a \tlcycliccode\ of length $n$ over $F_q$, $g(x)$ is the generator skew polynomial of $C,$ $ x^n -\lambda=h(x)g(x) ,$ $x^n - \theta^{ -k}(\lambda) = g(x) \hbar(x),$ $1-\lambda x^n = \hbar^\ast (x) \tilde{g}(x). $
    Suppose $C$ is not $\{0\}$ and $F_q^{n},$ then $C^\perp \subset C$ if and only if $\hbar^\ast (x) \hbar(x) $ is right divisible by $x^n - \theta^{-k}(\lambda) $ and $\lambda^{-1} = \lambda.$ Similarly, $C \subset C^\perp$ if and only if $g (x) \tilde{g}(x) $ is right divisible by $x^n - \lambda^{-1}$ and $\lambda^{-1} = \lambda.$
\end{theorem}
\begin{proof}
    Only the case $C^\perp \subset C$ will be proved, and the proof for $C \subset C^\perp$ is similar.

    Sufficiency. If $\hbar^\ast(x) \hbar(x) = \ell(x)(x^n-\theta^{-k}(\lambda)),$ then $\hbar^\ast(x) = \ell(x)g(x),$ so
    $$F_q[x;\theta]\left(\hbar^\ast(x) + \left\langle x^n-\lambda^{-1} \right\rangle \right) = F_q[x;\theta]\left(\hbar^\ast(x) + \left\langle x^n-\lambda\right\rangle \right) \subset F_q[x;\theta] \left(g(x) + \left\langle x^n-\lambda\right\rangle\right),$$
    Thus $\Phi \left(C^\perp\right) \subset \Phi(C),$ Therefore $C^\perp \subset C.$

    Necessity.    If $C^\perp \subset C,$ then $\hbar^\ast ( x) + \left\langle x^n-\lambda\right\rangle = \ell(x)\left(g(x) + \left\langle x^n-\lambda\right\rangle\right),$ therefore
    \begin{gather*}
        \hbar^\ast(x) = \ell(x)g(x) + u(x)(x^n - \lambda)  
                     = \left(\ell(x) + u(x)h(x)\right)g(x),
    \end{gather*}
    thus $g(x)$ right divides $\hbar^\ast(x),$ and
    \begin{gather*}
        \hbar^\ast (x) \hbar(x) =\left(\ell(x) + u(x)h(x)\right) g(x)\hbar(x)                             
                                = \left(\ell(x) + u(x)h(x)\right) \left(x^n - \theta^{-k}(\lambda)\right),
    \end{gather*}
    Thus $\hbar^\ast(x)\hbar(x) $ is right divisible by $x^n - \theta^{-k}(\lambda)$.

    We have already obtained that $g(x)$ right divides $\hbar^\ast(x),$ and that $\hbar^\ast(x)$ right divides $x^n - \lambda^{-1},$ by the property of the generator skew polynomial of a skew constacyclic code. Combine with $g(x)$ divides $x^n - \lambda$ to get $g(x)$ right divides $\lambda - \lambda^{-1},$ so $\lambda = \lambda^{-1}. $
\end{proof}

\section{Skew constacyclic codes over \texorpdfstring{$\fring$}{R}}
\subsection{Structure of \texorpdfstring{$\fring$}{R} and its automorphism}
Suppose $R$ is a finite commutative semisimple ring, then according to the Wedderburn-Artin Theorem, $R$ is isomorphic to the direct product of finite fields, that is, $R \cong \prod_{i=1}^s F_{q_i },$ where $F_{q_i}$ is a finite field. If we put the same direct product terms together, the above isomorphism can be written as
$$R \cong \prod_{j=1}^\ell \left(\prod_{k=1}^{t_j} F_{q_j} \right),$$
So we can see that the finite commutative semisimple ring $\prod_{i=1}^t F_q$ is the basic structural component of a general finite commutative semisimple ring, and the ring to be considered in this paper is exactly this kind of finite commutative semisimple ring.

Let $p$ be a prime number, $r$ be a positive integer, $q=p^r,$ denote a finite field with $q$ elements by $F_q$. In this paper, unless otherwise specified, let
\begin{equation*}
    \fring = \prod_{i=1}^t F_q = \left\{(x_1,x_2,\dots,x_t)\mid x_i \in F_q,1\leq i\leq t \right\},
\end{equation*}
the addition of two elements in the ring $\fring$ is the component addition, and the multiplication is the component multiplication. Let $e_1=(1,0,\dots,0),e_2=(0,1,0,\dots,0),\dots,e_t=(0,\dots,0,1)$, then the elements in  $\fring$ are of the form $x=x_1e_1 + \dots + x_te_t,$ where $x_i\in F_q.$ In the following, we always use $F_q$ linear sum of $e_1,e_2,\dots,e_t$ to represent elements in $\fring,$ instead of being written as vectors, and by convention $1e_i = e_i, 1\leq i\leq t.$

Under the above convention, for any $x=\sum_{i=1}^t x_ie_i,y=\sum_{i=1}^t y_ie_i\in \fring ,$  the addition and multiplication formulas are:
\begin{equation*}
    x+y = \sum_{i=1}^t (x_i+y_i)e_i, \qquad x\cdot y = \sum_{i=1}^t (x_iy_i)e_i.
\end{equation*}
The additive zero element of the ring $\fring$ is $0_{\fring} = 0e_1 + 0e_2 + \dots + 0e_t,$ the multiplication identity is $1_{\fring} = 1e_1 + 1e_2 + \dots + 1e_t.$ Write $0_{\fring}$ as $0.$ 

Obviously, the number of elements in $\fring$ is $|\fring| = q^t.$ The characteristic of $\fring$ is $\cchar(\fring) = \cchar(F_q) = p.$
Thus, for any $ x=\sum_{j=1}^t x_j e_j \in \fring, x^q=x_1^q e_1 + \dots + x_t^q e_t
    = x_1e_1 + \dots + x_te_t = x.$
According to $x\in U(\fring) $ if and only if $ x_j \in F_q^\ast, 1\leq j \leq t.$ We get \begin{equation*}
    U(\fring) \cong \underbrace{F_q^\ast \times \dots \times F_q^\ast}_{t\text{ terms}}
    \cong \underbrace{\mathbb{Z}_{q-1} \times \dots \times \mathbb{Z}_{q-1}}_{t\text{ terms}}.
\end{equation*}
and $|U(\fring)| = (q-1)^t.$

\begin{proposition}
    The ring $\fring$ is a principal ideal ring, with $2^t$ ideals and $t$ maximal ideals.
    The maximal ideal is $(1_{\fring}-e_j), 1\leq j\leq t.$
\end{proposition}
\begin{proof}
    Assuming that $I$ is an ideal of $\fring$, then $I = Ie_1 + \dots + Ie_t.$
    For any $ x=\sum_{i=1}^t x_ie_i \in \fring,$ where $x_i\in F_q,$ we have
    $xe_j = x_je_j \in F_qe_j,$ so $Ie_j \subset F_qe_j.$ If
    $Ie_j \neq \left\{0\right\},$ then there are non-zero elements $x_j \in F_q$
    such that $x_je_j \in Ie_j,$ so $e_j \in x_j^{-1}Ie_j = Ie_j,$
    So $F_qe_j = Ie_j.$

    Each $Ie_j, 1\leq j \leq t$ has two ways of choices, and different ways of choices correspond to different ideals,
    therefore there are $2^t$ ideals. If $I\neq \left\{0\right\} ,$ then the sum of $e_j$ such that $Ie_j\neq \left\{0\right\}$ is a generator of $I,$ which shows that $\fring$ is a principal ideal ring. The conclusion about the maximal ideal is obvious.
\end{proof}
\begin{definition}
    Suppose $R$ is a ring with identity $1_R,$ and $f$ is a mapping from $R$ to $R.$ If for any $x,y\in R,$ there is $f(x+y )=f(x)+f(y),\ f(xy) = f(x)f(y),\ f(1_R)=1_R,$ then $f$ is a ring homomorphism of $R.$ Moreover, if $f$ is a bijective homomorphism, then $f$ is said to be an automorphism of $R.$ Let $\Aut(R)$ be the group of automorphism of $R.$
\end{definition}

Before discussing $\Aut(\fring),$ some preparations need to be done.
\begin{definition}
    Let $e$ be an element in a commutative ring $R,$ if $e^2=e,$ then $e$ is an idempotent of $R.$ Let $e_1,e_2 \in R,$ if $e_1e_2=0,$ then $e_1$ is said to be orthogonal to $e_2.$ For an idempotent $e \in R,$ if $e=e_1+e_2,\ e_1^2=e_1 ,\ e_2^2=e_2,\ e_1e_2=0$ implies $e_1=0$ or $e_2=0,$ then $e$ is a primitive idempotent of $R.$
\end{definition}
\begin{lemma}
   The element $e_j,1\leq j\leq t,$
are pairwise orthogonal primitive idempotents of $\fring.$
\end{lemma}
\begin{proof}
    From the definition of the multiplication of elements in $\fring,$ $e_1,e_2,\dots,e_t$ are idempotents and pairwise orthogonal to each other, and then it is only necessary to show that these idempotents are primitive. The following gives the proof of $e_1$ is a primitive idempotent, and the proof that other idempotents are also primitive idempotents is similar.

If $e_1 = x+y,\ x^2=x,\ y^2 = y,\ xy=0,$ we need to show that $x=0$ or $y=0.$
Let $x=\sum_{j=1}^t x_je_j,\ y=\sum_{j=1}^t y_je_j,$
Then $e_1 = \sum_{j=1}^t (x_j+y_j)e_j,$
So $1 = x_1 + y_1,\ 0=x_j + y_j,\ 2\leq j\leq t.$
From $x^2 = x,\ y^2=y,\ xy=0$ we get
$x_j^2=x_j,\ y_j^2 = y_j,\ x_jy_j=0,\ 1\leq j\leq t.$
From the previous relations, we can get $x=e_1,\ y=0$ or $x=0,\ y=e_1.$
\end{proof}

\begin{lemma} \label{lem:morphismonbasis}
    The set $\{e_1,e_2,\dots,e_t\}$ is stable under the action of any automorphism of $\fring.$
\end{lemma}

\begin{proof}
    For any $\sigma \in Aut(\fring),$ we have
\begin{gather*}
    \sigma(e_1) + \dots + \sigma(e_t) = \sigma(e_1+\dots + e_t)
    =\sigma(1_{\fring}) = 1_{\fring}.\\
    \sigma(e_i)\sigma(e_j) = \sigma(e_ie_j) = \sigma(0) = 0, \quad 1\leq i,j\leq t.\\
    \sigma(e_i)\sigma(e_i) = \sigma(e_ie_i)=\sigma(e_i),\quad 1\leq i \leq t.
\end{gather*}
From the above equation it can be seen that $\sigma(e_1),\dots,\sigma(e_t)$
is also $t$ pairwise orthogonal idempotents of $\fring$ and sums to $1_{\fring}.$
Because $\sigma(e_1) \in \fring,$ $\sigma(e_1) = x_1e_1+\dots + x_te_t,$ where
$x_j\in F_q, 1\leq j \leq t.$
From $e_1^2 = e_1,$ then
$\sigma(e_1)\sigma(e_1) = \sigma(e_1),$
so $x_j^2 = x_j, 1\leq j\leq t,$ so $x_j=0$ or $x_j=1,$
$ 1\leq j\leq t.$ These $x_1,\dots,x_t$ cannot be all $0,$ otherwise
$e_1 = 0.$ and there is only one non-zero element in $x_1,\dots,x_t$, otherwise
$e_1=\sigma^{-1}(x_1e_1+\dots+x_te_t)$ decomposes into a sum of nontrivial orthogonal idempotents, which contradicts with $e_1$ is primitive. Similarly, $\sigma(e_j ) \in \{e_1,e_2,\dots,e_t\}, 1\leq j\leq t.$

Note again that $\sigma$ is bijective, so the action of $\sigma$ on $e_1,\dots,e_t$ is equivalent to a permutation on
$e_1,\dots,e_t.$
\end{proof}

\begin{remark}
    Denote $\theta(e_i) = e_{\bar{\theta}(i)},$ then $\bar{\theta}$ is naturally a permutation of $1,2,\dots,t,$ namely $\bar {\theta}\in S_t,$ where $S_t$ represents the symmetric group of $t$ elements, and $\bar{\theta}$ is uniquely determined by $\theta.$
\end{remark}

Use $F_q1_{\fring}$ to represent the set $\{a1_{\fring} \colon a\in F_q \}, $ then $F_q1_{\fring}$ is a subring of $\fring$ isomorphic to finite fields $F_q,$ but $F_q$ is not included in $\fring $ if $t>1.$ For any $\sigma \in \Aut(\fring),$ since $\sigma(1_{\fring} )=1_{\fring},$ we have
$\sigma|_{F_p1_{\fring}} = \id.$ For any $a\in F_p,$
$
    \sigma(ae_1) = \sigma(a1_{\fring}e_1)=\sigma(a1_{\fring})\sigma(e_1) = a1_{\fring} \sigma(e_1) = a\sigma(e_1).
$
So when $r=1,$ i.e., $q=p^r=p,$ once you know how $\sigma$ acts on $\{e_1,e_2,\dots,e_t\},$ the way $\sigma$ acts on $\fring$ is determined, and $\Aut(\fring) \cong
    S_t.$

For $r>1,$ i.e., $q\neq p$, you also need to know the action of $\sigma$ on $F_q1_{\fring}$ to completely determine $\sigma.$ In fact, if $\sigma $ is a injective ring homomorphism from $F_{q}1_{\fring}$ to $\fring,$ and the action of $\sigma$ on $e_1,\dots,e_t$ is permutation on
$e_1,\dots,e_t,$ then an automorphism of $\fring$ can be naturally obtained $\tilde{\sigma}: \alpha_1e_1 + \dots + \alpha_te_t \mapsto \sigma(\alpha_11_ {\fring})\sigma(e_1) + \dots + \sigma(\alpha_t1_{\fring})\sigma(e_t), \alpha_j\in F_q, 1\leq j \leq t.$

We can actually consider it a little simpler. In fact, in order to determine $\sigma,$ we only need to determine how $\sigma$ acts on each $F_qe_i,$ where $ 1\leq i\leq t. $ The general form of automorphism of $\fring$ is given below.

\begin{theorem} \label{thm:autofR}
    Let $\theta \in \Aut(\fring),$ then $\theta \left(\sum_{j=1}^t \alpha_j e_j \right) = \sum_{j=1}^t \psi_j(\alpha_j)\theta(e_j),$ where $\alpha_j \in F_q,$ $ \psi_j\in \Aut(F_q),$ $1\leq j\leq t.$
\end{theorem}

\begin{proof}
    Because 
    $$\theta(F_qe_j ) =\theta(F_qe_j e_j) \subset \theta(F_qe_j)\theta(e_j) \subset \fring \theta(e_j) = F_q\theta(e_j),$$
    and $\theta$ is bijective, we have $\theta(F_qe_j) = F_q\theta(e_j).$ So for any $a\in F_q,$ there is a unique $a^\prime \in F_q$ such that $\theta(ae_j) = a^\prime \theta(e_j), $ then we can define a mapping from $F_q$ to $F_q,$ i.e., $\psi_j: a\mapsto a^\prime.$
    For $\theta(ae_j) = \psi_j(a)\theta(e_j), \theta(be_j) = \psi_j(b)\theta(e_j), $ according to $\theta$ is a ring isomorphism of $\fring,$ we get $\psi_j (a+b)=\psi_j(a) + \psi_j(b),\psi_j(ab) = \psi_j(a)\psi_j(b), \psi_j(1)=1$ and $\psi_j$ is bijective, so $\psi_j \in \Aut(F_q).$
\end{proof}

\begin{remark}
    By imitating the above proof, we can determine automorphism of general finite commutative semisimple ring, and the automorphism group of $\fring$ to be given below can also be extended to general finite commutative semisimple ring.
\end{remark}

Let $G_1,G_2$ be a subset of $\Aut(\fring)$, where the element $\theta\in G_1$ satisfies $\theta(e_i) = e_i, 1\leq i \leq t,$ that is to say, $\theta$ acts on $e_1,\dots,e_t$ $\bar{\theta}$ as the identity map $\id.$ The element $\theta\in G_2$ satisfies $\theta\mid_{F_{q}1_{\fring}}= \id, $ that is to say that the elements in $G_2$ are the automorphisms of $\fring$ that keep the elements in $F_{q}1_{\fring}$ stable.

\begin{theorem}
    The $G_1,G_2$ defined above are subgroups of $\Aut(\fring),$ and
    \begin{equation*}
        G_1 \cong \underbrace{\mathbb{Z}_r\times \dots \times \mathbb{Z}_r }_{t \text{ terms}},\qquad G_2 \cong S_t, \qquad \Aut(\fring) /G_1 \cong G_2.
    \end{equation*}
\end{theorem}

\begin{proof}
    By Theorem \ref{thm:autofR} we can define a map from $G_1$ to $\prod_{i=1}^t \Aut(F_q)$: $\theta \mapsto (\psi_1,\dots,\psi_t ),$ we can directly verify that this is a group isomorphism, and then combine with $\Aut(F_q)\cong \mathbb{Z}_r$ to get the first isomorphism.

    For $\theta\in \Aut(\fring),$ from Lemma \ref{lem:morphismonbasis} we know that the action of $\theta$ on $e_1,\dots,e_t$ is equivalent to permutation on
$e_1,\dots,e_t,$ recall $\theta(e_i) = e_{\bar{\theta}(i)},$ then $\bar{\theta}$ is naturally a permutation of $1,2,\dots, t$. Define a map from $G_2$ to $S_t:$ $\theta \mapsto \bar{\theta},$ we can directly verify that this is a group isomorphism.

Consider the surjective homomorphism $f$ from $\Aut(\fring)$ to $S_t\colon \theta \mapsto \bar{\theta},$ then $\ker f = G_1,$ so $\Aut(\fring)/G_1 \cong S_t \cong G_2.$
\end{proof}

\begin{corollary}
    The number of elements in $\Aut(\fring)$ is $ |\Aut(\fring)| = r^t \times t!.$ \qed
\end{corollary}

Now we can give the structure of $\Aut(\fring).$

\begin{theorem}
     The automorphism group of $\fring$ is $\Aut(\fring) = G_1G_2. $
\end{theorem}

\begin{proof}
    First, $G_1G_2 \subset \Aut(\fring).$ By definition, $G_1 \cap G_2 = \{\id\},$ so $|G_1G_2| = |G_1||G_2|/|G_1\cap G_2| = |\Aut(\fring)|,$ so $\Aut(\fring) = G_1G_2. $
\end{proof}

\subsection{Decomposition of linear codes over \texorpdfstring{$\fring$}{R}}

In this section, we will establish a one-to-one correspondence between a linear code $C$ over $\fring$ and $t$ linear codes $C_1,C_2,\dots,C_t$ over $F_q,$ the initial idea can be traced back to Dougherty et al. \cite{dougherty1999self} on self-dual codes over $\mathbb{Z}_{2k}$. However, for the class of finite commutative semisimple ring we deal with here, this method is widely used. For example, Zhu et al. \cite{zhu2010some} showed how a linear code over $F_2[v]/(v^2-v)$ can be decomposed into linear codes over $F_2.$

First, the definition of linear codes over a finite ring with identity is given.

\begin{definition}
    Let $R$ be a finite ring with identity. If $C$ is a nonempty subset of free module $R^{n}$ of rank $n$ over $R,$ then $C$ is said to be a code over $R$ of length $n,$ and the elements in $C$ are called codewords. Let $x\in C,$ denote the number of nonzero components in $x$ by $\wt_H(x),$ call it the Hamming weight of $x.$ More specifically, if $C$ is a left $R$-submodule of $R^{n},$ then $C$ is said to be a linear code of length $n$ over $R.$
\end{definition}

We can also define inner product and dual code.

\begin{definition}
    Assume $R$ be a finite ring with identity. Let $u=(u_1,u_2,\dots,u_n),$ $v=(v_1,v_2,\dots,v_n) \in R^{n}, $ define the inner product of $u$ and $v$ by $u\cdot v = u_1v_1 + u_2v_2 + \dots + u_nv_n.$ If $C$ is a code of length $n$ over $R,$ define the dual code of $C $ is 
    \begin{equation*}
        C^{\perp} = \left\{z\in R^{n } : z\cdot x =0,\ \forall\ x\in C \right\}.
    \end{equation*}
\end{definition}

The ring $\fring$ can be viewed as a $t$ dimensional vector space over $F_q.$ In fact, there is an isomorphism of vector spaces. $\varphi: \fring \rightarrow F_q^{t},$ $\sum_{j=1}^t x_je_j
    \mapsto (x_1,x_2,\dots,x_t).$ Then there is a map induced by $\varphi$ from
$\fring^{n}$ to $F_q^{tn}$, and we still use $\varphi$ to represent it,
it maps \begin{equation*}
    \left(\sum_{j=1}^t x_{0j}e_j,\dots,\sum_{j=1}^tx_{n-1,j}e_j\right) \in \fring^{n}
\end{equation*}
 to
\begin{equation*}
    \left(x_{01},x_{11},\dots,x_{n-1,1},x_{02},x_{12},\dots,x_{n-1,2},\dots, x_{0t},x_{1t},\dots,x_{n-1,t}\right) \in F_q^{tn} .
\end{equation*}

\begin{remark} \label{expresscodesbymatrix}
    Observe matrix
    \begin{equation*}
        \begin{pmatrix}
            x_{01} & x_{02} & \dots & x_{0t} \\
            x_{11} & x_{12} & \dots & x_{1t} \\
            \vdots & \vdots & & \vdots \\
            x_{n-1,1} & x_{n-1,2} & \dots & x_{n-1,t}
        \end{pmatrix},
    \end{equation*}
    the preimage of $\varphi$ is the elements in the matrix arranged in rows, and the image is the elements in the matrix arranged in columns.
\end{remark}

\begin{definition}
    Define the weight of the codeword $x$ in $\fring^{n}$ as
    $\wt_{G}(x) = \wt_{H}(\varphi(x)),$
        define the distance between $x,y \in \fring^{n}$ as
    $\mathrm{d}_{G}(x,y) = \wt_{G}(x-y).$ If $C$ is a code over $\fring$ and it has at least two codewords, define the minimum distance of $C$ as $\mathrm{d}_G(C) = \min\{\mathrm{d}_G(x,y) \colon x,y\in C,\ x\neq y \}.$
\end{definition}

\begin{proposition}
    Consider $\fring^{n}$ and $F_{q}^{tn}$
    as vector spaces over $F_q,$ then $\varphi$ is an isomorphism of vector spaces and distance preserving.
\end{proposition}
\begin{proof}
    For any element in $r\in F_q,$ and for any
    \begin{align*} x & =\left(\sum_{j=1}^t x_{0j}e_j,\sum_{j=1}^tx_{1j}e_j,\dots,\sum_{j=1 }^tx_{n-1,j}e_j\right) \in \fring^{n}, \\
        y & =\left(\sum_{j=1}^t y_{0j}e_j,\sum_{j=1}^t y_{1j}e_j,\dots,\sum_{j=1}^t y_{ n-1,j}e_j\right)\in \fring^{n},\end{align*}
    by definition we have \begin{align*}
        \varphi(x+y) & = \left(x_{01}+y_{01},\dots,x_{n-1,1}+y_{n-1,1},\dots,x_{0t} +y_{0t},\dots,x_{n-1,t}+y_{n-1,t}\right) \\
                     & =\varphi(x) + \varphi(y).\\
        \varphi(rx) & = \left(rx_{01},\dots,rx_{n-1,1},\dots,rx_{0,t},\dots,rx_{n-1,t}\right)=r\varphi(x).
    \end{align*}
    So $\varphi$ is a $F_q$ linear map. Obviously, $\varphi$ is bijective, thus $\varphi$ is an isomorphic map. Because
    $$\mathrm{d}_{G}(x,y) = \wt_{G}(x-y) = \wt_{H}\left(\varphi(x-y)\right)=
        \wt_H\left(\varphi(x) - \varphi(y)\right) =
         \dhamming \left(\varphi(x),\varphi(y)\right),$$
    $\varphi$ is distance preserving.
\end{proof}

From the above proposition it follows:
\begin{corollary}
    If $C$ is a linear code of length $n$ over $\fring$ and its minimum distance is $\mathrm{d}_G,$
    then $\varphi(C)$ is a linear code over $F_q$ with length $tn$ and minimum distance $\mathrm{d}_G.$ \qed
\end{corollary}

\begin{definition}
    Let $A_1, A_2,\dots, A_t$ be non-empty sets, denote
    \begin{align*}
        A_1+A_2+\dots + A_t &= \left\{a_1 +a_2+ \dots + a_t \colon a_i\in A_i,\ 1\leq i \leq t\right\}, \\
        A_1\times A_2\times \dots \times A_t &= \left\{(a_1,a_2,\dots,a_t)\colon a_i\in A_i,\ 1\leq i\leq t\right\}.
    \end{align*}
    For any subset $C$ of $\fring^{n}$ and $1\leq j \leq t$ define
    \begin{equation*}
        C_j = \left\{(x_{0j},x_{1j},\dots,x_{n-1,j}) \in F_q^n : \exists\,\left(\sum_{j=1}^ tx_{0j}e_j,\dots,\sum_{j=1}^t x_{n-1,j}e_j\right) \in C\right\}.
    \end{equation*}
\end{definition}

\begin{definition}
    Let $G$ be a matrix over a finite ring $R,$ and every row vector of $G$ is a codeword in a linear code $C$ over $R.$ If the left $R $-module generated by the row vector of $G$ is $C,$ then the matrix $G$ is said to be a generator matrix of the linear code $C.$
\end{definition}

\begin{theorem} \label{decomposingcodes}      
        Let $C$ be a linear code of length $n$ over $\fring,$ then $C=C_1e_1 + \dots + C_te_t,$ where $C_j$ is a linear code of length $n$ over $F_q,$ $1\leq j \leq t.$
    
        Conversely, if $C_1,C_2,\dots,C_t$ are linear codes of length $n$ over $F_q,$ then $C_1e_1 + \dots + C_te_t$ is a linear code of length $n$ over $\fring.$ 

        If $G_j$ is a generator matrix of $C_j,$ $1\leq j\leq t,$ then
        \begin{equation*}
            \begin{pmatrix}
                G_1e_1 \\
                G_2e_2 \\
                \vdots \\
                G_te_t
            \end{pmatrix},
            \quad \begin{pmatrix}
                G_1 & & & \\
                    &G_2 && \\
                    & & \ddots & \\
                    & & & G_t
            \end{pmatrix}
        \end{equation*}
        are generator matrices of $C=C_1e_1 + \dots + C_te_t$ and $\varphi(C)= C_1\times \dots \times C_t$ respectively,
        $$\mathrm{d}_{G}(C) = \min\left\{\dhamming(C_j): 1\leq j\leq t\right\}.$$
\end{theorem}

\begin{proof}
    It can be verified according to the definition.
\end{proof}

\begin{theorem} \label{dualcodesoverR}
    Assuming that $C=C_1e_1 + \dots + C_te_t$ is a linear code over $\fring,$
    then $C^{\perp} = C_1^{\perp} e_1 + \dots + C_t^{\perp}e_t.$
    Moreover, $C\subset C^{\perp}$ if and only if
    $ C_j \subset C_j^{\perp}, 1\leq j \leq t.$ $C^{\perp}\subset C$ if and only if
    $ C_j^{\perp} \subset C_j, 1\leq j \leq t.$

\end{theorem}
\begin{proof}
    For any $x=x_1e_1 + \dots + x_te_t \in C^{\perp},$ where
    $x_j\in F_q^{n}, 1\leq j \leq t.$
    For any $ y_j \in C_j,$ by definition there exists $ y=y_1e_1 + \dots + y_te_t \in C,$
    so $x\cdot y = x_1\cdot y_1 e_1 + \dots + x_t\cdot y_t e_t = 0,$
    then $x_j\cdot y_j=0,$ so $x_j \in C_j^{\perp}, 1\leq j \leq t.$
    Thus, $C^{\perp} \subset C_1^{\perp}e_1 + \dots + C_t^{\perp}e_t.$

    For any
    $x=x_1e_1 + \dots + x_te_t\in C_1^{\perp}e_1 + \dots + C_t^{\perp}e_t,$
    where $x_j \in C_j^{\perp}, 1\leq j\leq t,$ for any
    $y_j \in C_j,$ we have $x_j\cdot y_j=0, 1\leq j \leq t.$
    So for any $ y=y_1e_1 + \dots + y_te_t \in C,$ there is $x\cdot y = 0,$
    therefore $x\in C^{\perp}.$ This shows that
    $C_1^{\perp}e_1 + \dots + C_t^{\perp}e_t \subset C^\perp.$

    In summary, $C^\perp = \left(C_1e_1 + \dots + C_te_t\right)^\perp=C_1^{\perp}e_1 + \dots + C_t^{\perp}e_t.$
    From this equation it immediately follows that the latter statements hold.
\end{proof}

\section{Characterization of skew constacyclic codes over \texorpdfstring{$\fring$}{R}}

\begin{definition}
    Let $\lambda\in U(\fring),$ $\theta\in \Aut(\fring),$ define a map from $\fring^n$ to $\fring^n$ $$\sigma_{\theta, \lambda}: (c_0,c_1,\dots,c_{n-1}) \mapsto (\lambda \theta(c_{n-1}),\theta(c_0),\dots,\theta(c_{n -2})).$$ Let $C$ be a linear code of length $n$ over $\fring.$ If for any $c \in C , $ we have $ \sigma_{\theta,\lambda }(c) \in C,$ then $C$ is said to be a $\theta$-$\lambda$-cyclic code of length $n$ over $\fring,$ if  $\theta$ and $\lambda$ is not specified, $C$ is also called a skew constacyclic code. When $\lambda =1,$ it is called a $\theta$-cyclic code or skew cyclic code.
\end{definition}

Similar to skew constacyclic code over $F_q$ considered earlier, let $\theta\in \Aut(\fring),$ define the skew polynomial ring over $\fring$
\begin{equation*}
    \fring[x;\theta] = \left\{a_0 +a_1x+\dots+a_kx^k\mid a_i \in \fring, 0\leq i \leq k\right\},
\end{equation*}
that is, the elements in $\fring[x;\theta]$ are elements in $\fring[x],$ but the coefficients are written on the left. Addition is defined in the usual way, and $ax^n \cdot (bx ^m) = a\theta^n(b)x^{n+m}, $ then define multiplication according to the associative and distributive laws. Similarly, it can be proved that in
$\fring[x;\theta]$ we can also do division with remainder on the right, as long as the divisor is monic. Use $\left\langle x^n-\lambda \right\rangle$ to represent the left $\fring[x;\theta]$-module $\fring[x;\theta](x^n-\lambda),$ then $\fring[x;\theta]/\left\langle x^n-\lambda \right\rangle $ is a left $\fring[x;\theta]$-module. Define a map $\Phi: $
\begin{align*}
    \fring^{n} & \rightarrow \fring[x;\theta]/\left\langle x^n-\lambda \right\rangle \\
    (c_0,c_1,\dots,c_{n-1}) & \mapsto c_0+c_1x+\dots+c_{n-1}x^{n-1}+\left\langle x^n-\lambda \right\rangle.
\end{align*}

For special cases of the following theorems see \citep[Theorem 2,][]{2012Onskewcyclic} by Abualrub et al., \citep[Theorem 3.5,][]{gao2013skew} by Gao,
\citep[Theorem 4.1,][]{islam2019note} by Islam et al., \citep[Theorem 3.1,][]{bag2019skew}, \citep[Theorem 2.4,][]{bag2020quantum} by Bag et al.

\begin{theorem}
    Let $\theta \in \Aut(\fring),$ $\lambda \in U(\fring),$ $C$ be a linear code of length $n$ over $\fring, $ 
    then $C$ is a \tlcycliccode\ over $\fring$ if and only if $\Phi(C)$ is a left $\fring [x;\theta]$-submodule of $\fring[x;\theta]/\left\langle x^n-\lambda \right\rangle$.
\end{theorem}

\begin{proof}
    Almost identical to the proof of Theorem \ref{thm:skewcodesoverFbymod}.
\end{proof}

Similar to Theorem \ref{dualcodesoverF}, the dual code of a skew constacyclic code over $\fring$ is still a skew constacyclic code.
 
\begin{theorem}
    Let $\theta \in \Aut(\fring),$ $\lambda \in U(\fring).$ If $C$ is a \tlcycliccode\ of length $n$ over $\fring,$ then $C ^\perp$ is a $\theta$-$\lambda^{-1}$-cyclic code of length $n$ over $\fring.$
\end{theorem}
\begin{proof}
    First, it can be directly verified that $C^\perp$ is a left $\fring$-module, that is, $C^\perp$ is a linear code of length $n$ over $\fring.$
    Then we just need to prove: for any $y\in C^\perp,$ there is $\sigma_{\theta,\lambda^{-1}}(y) \in C^\perp.$ For any $x\in C,$ since $C$ is a \tlcycliccode, there exists $\tilde{x} \in C,$ such that $ x = \sigma_{\theta,\lambda}(\tilde{x}).$ Note that $ \sigma_{\theta,\lambda^{-1}}(y) \cdot \sigma_{\theta,\lambda}(\tilde{x}) = \theta(y \cdot \tilde{x} ) =0,$ thus $\sigma_{\theta,\lambda^{-1}}(y) \in C^\perp.$
\end{proof}

\begin{remark}
    It is not difficult to see from the above proof that the above conclusion holds for skew constacyclic codes over finite commutative rings with identity.
\end{remark}

From the definition of skew constacyclic code, it can be seen that skew constacyclic code is related to the parameters $\theta$ and $\lambda,$ and the following discussion indeed shows that these two parameters are important to the structure of the skew constacyclic code. From the previous discussion, we already know the invertible element in $\fring$ and the form of the automorphism of $\fring,$ and know the linear code $C=C_1e_1+\dots+C_te_t$ over $\fring$ is completely determined by $C_1,\dots,C_t,$ with these preparations, we can give the characterization of \tlcycliccode\ over $\fring.$ Because the conclusions in some special cases can be written more clearly, and understanding the conclusions in special cases helps us understand the general case, here we first give the characterization of \tlcycliccode\ in two special cases, and finally give the most general conclusion.

\subsection{Special case one} \label{sec:specialone}

Let $\theta \in G_1,$ i.e., $\theta$ be an automorphism of $\fring$ that keeps $e_1,e_2,\dots,e_t$ stable, specifically $\theta(\alpha_1e_1 + \dots + \alpha_te_t) = \psi_1(\alpha_1)e_1 + \dots + \psi_t(\alpha_t)e_t,$ where $\psi_j \in \Aut(F_q),$ $\alpha_j\in F_q,1\leq j\leq t.$ In this subsection, it is always assumed that $\theta$ is as previously defined. If $\lambda = \lambda_1e_1 + \dots + \lambda_te_t$
is the invertible element in $\fring,$ where $\lambda_j \in F_q,1\leq j\leq t,$ then $\lambda$ is invertible is equivalent to $\lambda_j \neq 0$ for all $1\leq j\leq t.$
The \tlcycliccode\ over $\fring$ will be considered below. The main idea is to transform the problem into processing the $\psi_j$-$\lambda_j$-cyclic code over $F_q.$

For special cases of the following theorems, see \citep[Theorem 3,][]{gursoy2014construction} by Gursoy et al., \citep[Theorem 5,][]{gao2017skew} by Gao et al., \citep[Theorem 4.1,][]{shi2015skew} by Shi et al.,
\citep[Theorem 4.2 and Theorem 8.1,][]{islam2018skew}, \citep[Theorem 5.2,][]{islam2019note} by Islam et al., \citep[Theorem 3.3,][]{Ashraf2019Quantum} by Ashraf et al., \citep[Theorem 4.3,][]{bag2020quantum} by Bag et al.

\begin{theorem} \label{thm:mainofcase1}
    Let $\lambda = \lambda_1e_1 + \dots + \lambda_te_t \in U(\fring),$
    $C=C_1e_1 + \dots+ C_te_t$ is a linear code of length $n$ over $\fring.$ Then $C$ is a \tlcycliccode\ over $\fring$ if and only if $C_j$ is a $\psi_j$-$\lambda_j$-cyclic code over $F_q,$ $1\leq j\leq t.$
\end{theorem}
\begin{proof}
    Necessity. For any $x_j=(x_{0j},x_{1j},\dots,x_{n-1,j}) \in C_j, 1\leq j\leq t.$ We need to prove
    \begin{equation*}
        \left(\lambda_j\psi_j(x_{n-1,j}),\psi_j(x_{0j}),\dots,\psi_j(x_{n-2,j})\right) \in C_j.
    \end{equation*}
    From $x_j \in C_j$, $x=x_1e_1 + \dots + x_te_t\in C,$
    then
    \begin{gather*}    \left(\lambda\theta\left(\sum_{j=1}^t x_{n-1,j}e_j\right),\theta\left(\sum_{j=1}^tx_{0j}e_j\right),\dots,\theta\left(\sum_{j=1}^tx_{n-2,j}e_j\right)\right) \in C,\\
    \left(\sum_{j=1}^t \lambda_j \psi_j(x_{n-1,j})e_j,\sum_{j=1}^t\psi_j(x_{0j})e_j,\dots, \sum_{j=1}^t\psi_j(x_{n-2,j})e_j\right)\in C,    \end{gather*}
    so $\left(\lambda_j\psi_j(x_{n-1,j}),\psi_j(x_{0j}),\dots,\psi_j(x_{n-2,j})\right) \in C_j .$

    Sufficiency. For any
    \begin{gather*}
        \left(\sum_{j=1}^t x_{0j}e_j,\sum_{j=1}^tx_{1j}e_j,\dots,\sum_{j=1}^tx_{n-1,j }e_j\right)\in C,
        \shortintertext{we need to proof}
        \left(\lambda\theta\left(\sum_{j=1}^tx_{n-1,j}e_j\right),\theta\left(\sum_{j=1}^tx_{0j}e_j\right),\dots,\theta\left(\sum_{j=1}^tx_{n-2,j}e_j\right)\right) \in C.
    \end{gather*}
    For $(x_{0j},x_{1j},\dots,x_{n-1,j}) \in C_j,$ we have $\left(\lambda_j\psi_j(x_{n-1,j}),\psi_j(x_{0j}),\dots,\psi_j(x_{n-2,j})\right) \in C_j,$
    then
    \begin{equation*}
        \left(\sum_{j=1}^t \lambda_j \psi_j(x_{n-1,j})e_j,\sum_{j=1}^t\psi_j(x_{0j})e_j,\dots, \sum_{j=1}^t\psi_j(x_{n-2,j})e_j\right) \in C.
    \end{equation*}
\end{proof}

The skew polynomials can used to describe skew constacyclic codes over $\fring,$ mainly using the relevant conclusions of the skew constacyclic codes over $F_q$ obtained before. The special cases of the following theorem can be found in \citep[Theorem 4 and Theorem 5,][]{gursoy2014construction}, \citep[Theorem 6,][]{gao2017skew} by Gao et al.,
\citep[Theorem 4.2 and Theorem 4.3,][]{shi2015skew} by Shi et al., \citep[Theorem 4.3 and Theorem 8.2, ][]{islam2018skew}, \citep[Theorem 5.5,][]{islam2019note} by Islam et al.,  \citep[Theorem 3.4,][]{Ashraf2019Quantum} by Ashraf et al., \citep[Theorem 4.6,][]{bag2020quantum} by Bag et al..

\begin{theorem}\label{characterizedbypolynomial}
    Let $C=C_1e_1 + \dots + C_te_t$ be a \tlcycliccode\ of length $n$ over $\fring,$ and the generator skew polynomial of $C_j$ is $g_j(x), 1\leq j \leq t, $ then
    \begin{align*}
        \Phi(C) & = \fring[x;\theta]\left(g_1(x)e_1+\left\langle x^n-\lambda \right\rangle,\dots,g_t(x)e_t + \left\langle x ^n-\lambda \right\rangle \right) \\
            & = \fring[x;\theta] \left(g_1(x)e_1 + \dots + g_t(x)e_t + \left\langle x^n-\lambda \right\rangle \right).
    \end{align*}
    Where $g_1(x)e_1 + \dots + g_t(x)e_t$ is a right factor of $x^n - \lambda.$
    \begin{equation*}
        \left|C\right| = \prod_{j=1}^t \left|C_j\right| = q^{\sum_{j=1}^t \left(n-\deg g_j(x)\right)}.
    \end{equation*}
\end{theorem}

\begin{proof}
    Since $ e_j \left(g_1(x)e_1 + \dots + g_t(x)e_t\right) = g_j(x)e_j,$ the second equality is naturally true. Because $g_j(x)e_j + \left\langle x^n-\lambda \right\rangle$ corresponds to an element in $C_je_j,$ an element in $C_je_j$ is naturally an element in $C,$ so $g_j(x)e_j +\left\langle x^n-\lambda \right\rangle \in \Phi(C), 1\leq j\leq t,$ thus
    \begin{equation*}
        \Phi(C) \supset \fring[x;\theta]\left(g_1(x)e_1+\left\langle x^n-\lambda \right\rangle,\dots,g_t(x)e_t + \left\langle x ^n-\lambda \right\rangle\right).
    \end{equation*}
    For any $ c(x) + \left\langle x^n-\lambda \right\rangle \in \Phi(C),$ the coefficients of $c(x) $ are in $\fring,$ and it can be written as the $F_q$ linear sum of $e_j .$ Then $c(x) = c_1(x)e_1 + \dots + c_t(x)e_t,$ where $c_j(x) \in F_q[x;\theta],$
    $c_1(x)$ corresponds to an element in $C_1,$ then $c_1(x) = q_1(x)g_1(x) + k(x)(x^n-\lambda_1),$ so $c_1(x )e_1 = \left(q_1(x)g_1(x) + k(x)(x^n-\lambda)\right)e_1.$
    Therefore,
    \begin{equation*}
        c(x) +\left\langle x^n-\lambda \right\rangle = q_1(x)g_1(x)e_1 + \dots + q_t(x)g_t(x)e_t + \left\langle x^n -\lambda \right\rangle ,
    \end{equation*}
    and $c(x)+\left\langle x^n-\lambda \right\rangle \in \fring[x;\theta]\left(g_1(x)e_1+\left\langle x^n-\lambda \right\rangle,\dots,g_t(x)e_t + \left\langle x^n-\lambda \right\rangle\right).$

    Since $h_j(x)g_j(x) = x^n - \lambda_j,$ then $h_j(x)g_j(x)e_j = (x^n-\lambda)e_j,$ $ 1\leq j\leq t,$ thus
    \begin{equation*}
        \left(\sum_{j=1}^th_j(x)e_j\right)\left(\sum_{j=1}^t g_j(x)e_j\right) = \sum_{j=1}^t h_j (x)g_j(x)e_j = x^n-\lambda.
    \end{equation*}
\end{proof}

The above conclusion tells us that if $C$ is a \tlcycliccode\ of length $n$ over $\fring,$ then $\Phi(C)$ can be generated by only one element.
From the skew polynomial characterization of $C,$ the skew polynomial characterization of $C^\perp$ can be obtained. The special case of the following corollary can be found in \citep[Corollary 7,][]{gursoy2014construction} by Gursoy et al.,
\citep[Corollary 4.4,][]{shi2015skew} by Shi et al., \citep[Corollary 4.4 and Corollary 8.4,][]{islam2018skew} by Islam et al.,
\citep[Corollary 3.7,][]{Ashraf2019Quantum} by Ashraf et al.,  \citep[Corollary 4.8,][]{bag2020quantum} by Bag et al.. They assumed $\ord(\theta) \mid n$ and $\theta(\lambda) = \lambda$.

\begin{corollary}
    If $C$ is a \tlcycliccode\ of length $n$ over $\fring,$ $\theta \in \Aut(\fring),$ $\lambda = \sum_{j=1}^t \lambda_je_j \in U(\fring),$ $C=C_1e_1 + \dots + C_te_t,$ then $C^\perp$ is a $\theta$-$\lambda^{-1}$-cyclic code of length $n$ over $\fring.$ For $1\leq j\leq t,$ if $C_j$ is a $\psi_j$-$\lambda_j$-cyclic code of length $n$ over $F_q,$ $g_j(x)$ is the generator skew polynomial of $C_j$ and $h_j(x)g_j(x) = x^n - \lambda_j,$ then $C_j^\perp$ is a $\psi_j$-$\lambda_j$-cyclic code of length $n$ over $F_q$ and the generator skew polynomial of $C_j^{\perp}$ is $\hbar_j^\ast (x)$ left multiplied by the inverse of its coefficient of leading term, and $\left|C^\perp\right| = q^{\sum_{j=1}^t \deg g_j(x)},$
    \begin{align*}
        \Phi \left(C^\perp\right) & = \fring [x;\theta]\left(\hbar_1^\ast(x) e_1 + \left\langle x^n-\lambda^{-1} \right\rangle ,\dots,\hbar_t^\ast(x)e_t + \left\langle x^n-\lambda^{-1} \right\rangle \right) \\
                               & = \fring[x;\theta]\left(\hbar_1^\ast(x)e_1 + \dots + \hbar_t^\ast(x)e_t+ \left\langle x^n-\lambda^{-1} \right\rangle \right).
    \end{align*}
\end{corollary}
\begin{proof}
    It follows from Theorem \ref{dualcodesoverR}, Theorem \ref{dualcodesoverF}, Theorem \ref{polynomialofdualcodes}, Theorem \ref{characterizedbypolynomial}.
\end{proof}

Using Theorem \ref{containingcondition}, the necessary and sufficient conditions that a \tlcycliccode\ and its dual code over $\fring$ have mutual containment relation can also be obtained, which is omitted here.

\subsection{Special case two}

For $\theta \in \Aut(\fring),$ we have showed that $\theta$ is determined by the way it acts on $F_q1_{\fring}$ and the set of primitive idempotent
$
    \{e_1,e_2,\dots,e_t\}.
$
The effect of $\theta$ on the set $\{e_1,e_2,\dots,e_t\}$ in the case considered above is trivial. In the case now to be considered, the action of $\theta$ on $F_q1_{\fring}$ is trivial, and more specifically, in this section, unless otherwise specified, it is always assumed
$$ \theta: \fring \rightarrow \fring,\quad \alpha_1e_1 + \dots +\alpha_t e_t \mapsto \alpha_1e_2 + \dots + \alpha_{t-1}e_t + \alpha_te_1, $$
that is to say, $\theta\in G_2$ and $\theta$ corresponds to $\bar{\theta}=(1,2,\dots,t)$ in the symmetric group $S_t.$  In addition, we only consider $\theta$-cyclic codes $C$ of length $n$ over $\fring.$ From Theorem \ref{decomposingcodes}, determine $C =C_1e_1 + \dots + C_te_t$ is equivalent to determine $C_i,\ 1\leq i\leq t.$

For the convenience of notation, we define a map first.

\begin{equation*}
    \sigma_{\theta}: \fring^{n} \rightarrow \fring^{n},\quad (c_0,c_1,\dots,c_{n-1}) \mapsto (\theta(c_{n-1 }),\theta(c_0),\dots,\theta(c_{n-2})).
\end{equation*}

The following theorem generalizes Gao's \cite[Theorem 3.7,][]{gao2013skew}.
\begin{theorem} \label{thm:specialcase}
    If $(t,n) =1,$ $C = C_1e_1 + \dots + C_te_t$ is a linear code of length $n$ over $\fring.$ Then $C$ is a $\theta$-cyclic code of length $n$ over $\fring$ if and only if $C_1 = C_2 = \dots = C_t$ and $C_1$ is a cyclic code of length $n$ over $F_q.$
\end{theorem}

\begin{proof}
    Sufficiency. For all $x = x_1e_1 + \dots + x_te_t \in C,$ where $x_i \in C_i, 1\leq i \leq t.$ We have $\sigma_{\theta}(x) = \rho(x_1)e_{\bar{\theta}(1)}+ \dots + \rho(x_t)e_{\bar{\theta}(t)} = \rho(x_t)e_1 + \rho(x_1) e_2 + \dots + \rho(x_{t-1})e_t \in C.$

    Necessity. Because $(t,n)=1,$ there exist $u,v \in \mathbb{Z},$ such that $ut + vn =1,$ then for any $k\in \mathbb{Z},$ $(u+kn)t = 1+ (kt-v)n,$ so there are positive integers $a,b$ such that $at = 1+bn,$ so that for any $ x = x_1e_1 + \dots + x_te_t \in C,$ where $x_i \in C_i, 1\leq i \leq t,$ we have
    \begin{align*}
        \sigma_{\theta}^{at}(x) & = \rho^{at}(x_1)e_1 + \dots + \rho^{at}(x_t)e_t \\
                                & = \rho^{1+bn}(x_1)e_1 + \dots + \rho^{1+bn}(x_t)e_t \\
                                & = \rho(x_1)e_1 + \dots + \rho(x_t)e_t \in C,
    \end{align*}
    Thus $C_i$ is a cyclic code of length $n$ over $F_q,$ $1\leq i\leq t.$

    There are also positive integers $a^\prime,b^\prime$ such that $a^\prime n = 1+b^\prime t,$ so
    \begin{equation*}
        \sigma_{\theta}^{a^\prime n} (x) = x_te_1 + x_1e_2 + x_2e_3 + \dots + x_{t-1}e_t \in C,
    \end{equation*}
    thus for any $x_1 \in C_1,$ there is $x_1 \in C_2,$ then $C_1 \subset C_2.$ Similarly, $C_1\subset C_2 \subset C_3 \subset \dots \subset C_t \subset C_1,$ so $C_1 = C_2 = \dots = C_t.$
\end{proof}

\begin{remark}
    This theorem can be regarded as a special case of Theorem \ref{thm:generalcase} to be proved below, but for the reader to understand the following result better, this result is shown here first. Using Theorem \ref{thm:mainofcase1} we can obtain: under the conditions of the above theorem, $\theta$-cyclic code over $\fring$ must be a cyclic code over $\fring.$
\end{remark}

\begin{corollary} \label{cor:specialcase_coprime}
    If $(t,n) =1,$ $C = C_1e_1 + \dots + C_te_t$ is a linear code of length $n$ over $\fring.$ If $C$ is a $\theta$-cyclic code of length $n$ over $\fring,$ then $C^{\perp} = C_1^{\perp} e_1 + \dots + C_t^{\perp}e_t$ is also a $\theta$-cyclic code of length $n$ over $\fring.$ Moreover, $C^\perp = C$ if and only if $C_1^\perp = C_1.$ Let $x^n -1 = p_1^{k_1}(x) \dots p_s^{k_s}(x) $ is the decomposition of $x^n -1$ in $F_q[x]$, where $k_i$ is a positive integer and $p_i(x)$ is a monic irreducible polynomial in $F_q[x],$ $1\leq i \leq s,$ then the number of $\theta$-cyclic codes of length $n$ over $\fring$ is $(1+k_1 )\dots (1+k_s).$
\end{corollary}
\begin{proof}
    The first two assertions can be obtained from Theorem \ref{dualcodesoverR} and Theorem \ref{thm:specialcase}. Since cyclic codes over $F_q$ one to one corresponds to ideals of  $F_q[x]/(x^n-1),$ we get the conclusion about the number of $\theta$-cyclic codes.
\end{proof}

The following theorem generalizes Gao's \cite[Theorem 3.3,][]{gao2013skew}.

\begin{theorem} \label{thm:generalcase}
    If $(t,n) =\ell,$ $C = C_1e_1 + \dots + C_te_t$ is a linear code of length $n$ over $\fring.$ Then $C$ is a $\theta$-cyclic code of length $n$ over $\fring$ if and only if $C_i = C_{i+\ell} = \dots = C_{i+t-\ell}, 1\leq i \leq \ell, $ $\rho^\ell(C_1) = C_1,$ and $C_2 = \rho(C_1), C_3 = \rho^2(C_1),\dots,C_\ell = \rho^{\ell-1} (C_1).$ The number of $\theta$-cyclic codes of length $n$ over $\fring$ is equal to the number of quasi cyclic codes of length $n$ with index $\ell$ over $F_q.$
\end{theorem}

\begin{proof}
    Sufficiency. For any $x = x_1e_1 + \dots + x_te_t \in C,$ where $x_i \in C_i, 1\leq i \leq t.$ $\sigma_{\theta}(x) = \rho( x_t)e_1 + \rho(x_1)e_2 + \dots + \rho(x_{t-1})e_t \in C.$

    Necessity. Because $(t,n)=\ell,$ there is a positive integer $a^\prime,b^\prime$ such that $a^\prime n = \ell+b^\prime t,$ then for any $x = x_1e_1 + \dots + x_te_t \in C,$ where $x_i \in C_i, 1\leq i \leq t,$ we have
    \begin{equation*}
        \sigma_{\theta}^{a^\prime n} (x) = x_1e_{1+\ell} + x_2e_{2+\ell} + \dots +x_{t-\ell}e_t+ x_{t-\ ell+1}e_1+\dots+ x_{t}e_{\ell} \in C,
    \end{equation*}
    so $C_i\subset C_{i+\ell} \subset C_{i+2\ell} \subset \dots \subset C_{i+t-\ell} \subset C_i, 1\leq i \leq \ell.$ Thus $C_i = C_{i+\ell} = \dots = C_{i+t-\ell}, 1\leq i \leq \ell.$

    There are positive integers $a,b$ such that $at = \ell+bn,$ thus
    \begin{align*}
        \sigma_{\theta}^{at} (x) & = \rho^{at}(x_1)e_1 + \dots + \rho^{at}(x_t)e_t \\
                                 & = \rho^{\ell+bn}(x_1)e_1 + \dots + \rho^{\ell+bn}(x_t)e_t \\
                                 & = \rho^\ell (x_1)e_1 + \dots + \rho^\ell (x_t)e_t \in C,
    \end{align*} 
    for any $ x_i \in C_i, $ we have $\rho^\ell (x_i) \in C_i,$ so $\rho^\ell (C_i) =C_i, 1\leq i \leq t.$

    From $\sigma_{\theta}(x) = \rho(x_t)e_1 + \rho(x_1)e_2 + \dots + \rho(x_{t-1})e_t \in C,$ we know $\rho( C_1) \subset C_2,$ $\rho(C_2) \subset C_3,$ $\dots,$ $\rho(C_t) \subset C_1,$
    Thus $C_1 = \rho^\ell(C_1) \subset \rho^{\ell-1}(C_2) \subset \rho^{\ell-2}(C_3) \subset \dots \subset \rho^{ 2}(C_{\ell-1}) \subset \rho(C_\ell) \subset C_{\ell+1} = C_1.$ So $C_2 = \rho(C_1), C_3 = \rho^2( C_1),\dots,C_\ell = \rho^{\ell-1}(C_1).$

    According to the above analysis, if $C = C_1e_1 + \dots + C_te_t$ is a $\theta$-cyclic code of length $n$ over $\fring,$ then $\rho^\ell(C_1) = C_1 ,$ and $C_j$ is determined by $C_1$ for $2\leq j \leq t.$ Conversely, if $\rho^\ell(C_1) = C_1,$ and $C_j, 2\leq j \leq t,$ is determined according to the above relationship. Then $C = C_1e_1 + \dots + C_te_t$ is a $\theta$-cyclic code of length $n$ over $\fring.$ Therefore, we find that $C$ is determined by $C_1$ satisfying $\rho^\ell(C_1) = C_1,$ from which the final conclusion about the number of $\theta$-cyclic codes can be obtained.
\end{proof}

For general $\theta \in G_2,$ let its corresponding element in $S_t$ be $\bar{\theta},$ then $\bar{\theta}$ can be expressed as the product of some disjoint cycles, if $C$ is a $\theta$-cyclic code over $\fring,$ then $C_1,C_2,\dots,C_t$ are partitioned according to cycles, each part is essentially determined by one code, it can also be proved that each $C_i$ is a quasi-cyclic code over $F_q.$ Since this result can be obtained from the following general result, we do not write a separate proof.

\subsection{General case}

Let $\theta\in \Aut(\fring),$ $\theta$ act on $F_q1_{\fring}$ as $\theta(a1_{\fring}) = \psi_1(a)\theta( e_1)+\dots+\psi_t(a)\theta(e_t), \psi_j \in \Aut(F_q), 1\leq j \leq t.$ And $\theta$ corresponds to $\bar{\theta}$ in the symmetric group $S_t.$ That is to say, for any $\sum_{j=1}^ta_je_j\in \fring,$ \ $\theta\left(\sum_{j=1}^t a_je_j\right) = \sum_{j=1}^t \psi_j(a_j) e_{\bar{\theta}(j)}.$

A general characterization of \tlcycliccode\ over $\fring$ is given below.
\begin{theorem}[Characterization of skew constacyclic code] \label{thm:Generalcase}
    Let $\theta\in \Aut(\fring),$ $\lambda=\lambda_1e_1+\dots+\lambda_te_t,$ where $\lambda_j\in F_q^\ast, 1\leq j \leq t.$ If $C= C_1e_1+\dots+C_te_t$ is a linear code of length $n$ over $\fring,$ then $C$ is a \tlcycliccode\ over $\fring$ if and only if $\rho_{\psi_{j}, \lambda_{\bar{\theta}(j)}}(C_j)= C_{\bar{\theta}(j)}, 1\leq j \leq t.$
\end{theorem}
\begin{proof}
    For any $ x=x_1e_1 + \dots + x_te_t \in C,$ it can be calculated directly that
    \begin{align*}
        \sigma_{\theta,\lambda}(x) & = \left( \lambda\theta\left( \sum_{j=1}^t x_{n-1,j}e_j \right),\theta\left( \sum_{j=1}^t x_{0j}e_j \right),\dots,\theta\left( \sum_{j=1}^t x_{n-2,j}e_j \right) \right) \\
                                        & =\left( \lambda \sum_{j=1}^t \psi_{j}\left(x_{n-1,j}\right)e_{\bar{\theta}(j)}, \sum_ {j=1}^t \psi_{j}\left(x_{0j}\right)e_{\bar{\theta}(j)}, \dots,\sum_{j=1}^t \psi_{ j}\left(x_{n-2,j}\right)e_{\bar{\theta}(j)} \right) \\
                                        & = \left( \sum_{j=1}^t \lambda_{\bar{\theta}(j)} \psi_{j}\left(x_{n-1,j}\right)e_{\bar {\theta}(j)}, \sum_{j=1}^t \psi_{j}\left(x_{0j}\right)e_{\bar{\theta}(j)}, \dots,\ sum_{j=1}^t \psi_{j}\left(x_{n-2,j}\right)e_{\bar{\theta}(j)} \right), \\
        \sigma_{\theta,\lambda}(x_je_j) & = \left( \lambda\theta\left(x_{n-1,j}e_j\right),\theta\left(x_{0j}e_j\right) ,\dots,\theta\left(x_{n-2,j}e_j\right) \right) \\
                                        & = \left( \lambda_{\bar{\theta}(j)} \psi_{j}\left(x_{n-1,j}\right)e_{\bar{\theta}(j)}, \psi_{j}\left(x_{0j}\right)e_{\bar{\theta}(j)}, \dots, \psi_{j}\left(x_{n-2,j}\right) e_{\bar{\theta}(j)} \right) \\
                                        & = \rho_{ \psi_{j},\lambda_{\bar{\theta}(j)} }(x_j)e_{\bar{\theta}(j)}.
    \end{align*}

    Necessity. If $C$ is a \tlcycliccode, then for any $x_j=(x_{0j},x_{1j},\dots,x_{n-1,j}) \in C_j,$ naturally $x_je_j \in C,$ thus $\sigma_{\theta,\lambda}(x_je_j) \in C,1\leq j \leq t.$
    So for any $ x_j \in C_j,$ $\rho_{\psi_{j},\lambda_{\bar{\theta}(j)}}(x_j)\in C_{\bar{\theta}(j )}, 1\leq j \leq t.$
    Thus $\rho_{\psi_{j},\lambda_{\bar{\theta}(j)}}(C_j)\subset C_{\bar{\theta}(j)}, 1\leq j \leq t .$ Since $\rho_{\psi_j,\lambda_{\bar{\theta}(j)}}$ is injective, so $\left|C_j\right| \leq \left|C_{\bar{\theta}(j)}\right| \leq \left|C_{\bar{\theta}^2(j)}\right|\leq \dots \leq \left|C_j\right| ,$ we get $\rho_{\psi_{j,\lambda_{\bar{\theta}(j)}}}(C_j)= C_{\bar{\theta}(j)}, 1\leq j \leq t.$

    Sufficiency. It can be directly verified that for any $ x=x_1e_1 + \dots + x_te_t \in C,$ there is $\sigma_{\theta,\lambda}(x) \in C.$
\end{proof}

As an application of the above theorem, we can obtain a general characterization of skew cyclic codes over $\fring.$

\begin{theorem}
    Let $\theta\in \Aut(\fring),$ $C=C_1e_1+\dots+C_te_t$ be a linear code of length $n$ over $\fring,$ then $C$ is a $\theta$-cyclic code over $\fring$ if and only if $\rho_{\psi_{j}}(C_j)= C_{\bar{\theta}(j)},\ 1\leq j \leq t.$ \qed
\end{theorem}

\begin{remark}
    This is a more general conclusion than Theorem \ref{thm:specialcase} and Theorem \ref{thm:generalcase}, but since $\theta$ is more general, the resulting necessary and sufficient conditions are not as clear as the previous ones.
\end{remark}
\section{Image of homomorphism}
The reason for studying linear codes over finite commutative rings is: some nonlinear codes with good parameters over binary field can be regarded as the Gray images of linear codes over $\mathbb{Z}_4$ \cite{ Z4linear}, thus promoting the understanding of the original nonlinear codes. Similar to the previous results, we will establish relationship between linear codes over $\fring$ and linear codes over $F_q$ by defining homomorphisms. The main results we get are: homomorphism image of a linear code over $\fring$ is a matrix product code over $F_q,$ some optimal linear codes over $F_q$ can also be viewed as images of the homomorphisms we have defined.

\subsection{Isomorphism}

We previously defined an isomorphic map $\varphi$ from $\fring^n$ to $F_q^{tn}$ and found that the linear code $C$ over $\fring$ is completely determined by the linear code $C_1,C_2,\dots,C_t$ over $F_q.$ About $\varphi(C)$ we have the following conclusion.

\begin{proposition}
    If $C=C_1e_1 + \dots + C_te_t$ is a linear code of length $n$ over $\fring,$ where the parameters of $C_i$ are $[n,k_i,d_i],$ then $\varphi(C) = C_1\times \dots \times C_t =\{(x_1,\dots,x_t) \in F_q^{tn}: x_i \in C_i, 1\leq i\leq t\}$ is a linear code over $F_q$ with parameters $[tn,k_1+\dots + k_t ,\min\{d_1,\dots,d_t\}].$ \qed
\end{proposition}

\begin{example}
    Let $C_1=C_2=C_3=C_4$ be a cyclic code of length $3$ generated by $x-1$ over $F_2,$ then their parameters are $[3,2,2],$ then $C=C_1e_1 + C_2e_2$ is a cyclic code of length $3$ over $F_2\times F_2,$ and $\varphi(C)$ is a linear code over $F_2$ with parameters $[6,4,2].$ Look up the table \cite{Grassl:codetables}, we see that it is an optimal linear code. Similarly, $\varphi(C_1e_1 + C_2e_2+C_3e_3)$ is a linear code over $F_2$ with parameters $[9,6, 2],$ which is also an optimal linear code. However, $\varphi(C_1e_1 + C_2e_2+C_3e_3+C_4e_4)$ is a linear code with parameters $[12,8,2]$ over $F_2$, which is not an optimal linear code.
\end{example}

Because the length and dimension of $\varphi(C)$ is larger than the sum of $C_1, C_2, \dots, C_t$ respectively, but the minimum distance does not become larger, so in general $\varphi(C)$ will not have good parameters. But we can give linear codes with good parameters over $F_q$ by constructing other isomorphism.

Let $M=(m_{ij})_{t\times t}$ be an invertible matrix over $F_q,$ consider the mapping $\eta_M$ from $\fring $ to $F_q^t :$
\begin{equation*}
    \sum_{k=1}^t a_ke_k \mapsto (a_1,a_2,\dots,a_t)M=\left( \sum_{k=1}^t a_km_{k1}, \sum_{k=1}^t a_km_{k2},\dots,\sum_{k=1}^t a_km_{kt}, \right) .
\end{equation*}
Obviously, $\eta_M$ is an isomorphism of vector spaces over $F_q,$ and induces a mapping from $\fring^n$ to $F_q^{tn},$ still denoted by $\eta_M$, it maps
\begin{equation*}
    x = \left(x_0,x_1,\dots,x_{n-1}\right)=\left(\sum_{k=1}^t x_{0k}e_k,\sum_{k=1}^t x_ {1k}e_k,\dots,\sum_{k=1}^t x_{n-1,k}e_k\right)
\end{equation*}
to
\begin{equation*}
    \begin{gathered}
        \biggl( \sum_{k=1}^t x_{0k}m_{k1},\sum_{k=1}^t x_{1k}m_{k1},\dots,\sum_{k=1}^ t x_{n-1,k}m_{k1},\phantom{\biggr)}\\
        \phantom{\biggl(} \sum_{k=1}^t x_{0k}m_{k2},\sum_{k=1}^t x_{1k}m_{k2},\dots,\sum_{k =1}^t x_{n-1,k}m_{k2} ,\phantom{\biggr)} \\
        \dots, \\
        \phantom{\biggl(}\sum_{k=1}^t x_{0k}m_{kt},\sum_{k=1}^t x_{1k}m_{kt},\dots,\sum_{k =1}^t x_{n-1,k}m_{kt} \biggr),
    \end{gathered}
\end{equation*}
which is
\begin{equation*}
    \eta_M(x) =\varphi(x) \left(M\otimes E_n \right),
\end{equation*}
where $E_n$ represents the identity matrix of order $n$ over $F_q,$ and $M\otimes E_n$ represents the tensor product of $M$ and $E_n.$
It is not difficult to find that $\eta_M$ is an isomorphism from the vector space $\fring^n$ over $F_q$ to $F_q^{tn},$ and $\varphi$ defined previously is actually $\eta_{E_t}.$ Let $x_k = \left(x_{0k},x_{1k},\dots,x_{n-1,k}\right), 1\leq k \leq t,$ then $\varphi( x) = (x_1,x_2,\dots,x_t),$
\begin{equation*}
    \begin{aligned}
        \eta_M (x) & = \eta_M (x_1e_1 + x_2e_2 + \dots + x_te_t) \\
                   & = \left( \sum_{k=1}^t m_{k1}x_k, \sum_{k=1}^t m_{k2}x_k,\dots,\sum_{k=1}^t m_{kt }x_k \right).
    \end{aligned}
\end{equation*}

\begin{lemma} \label{lem:dualofimage}
    Let $C$ be a linear code of length $n$ over $\fring,$ then $\varphi\left(C^\perp\right)=\varphi(C)^\perp,$ and $C= C^\perp$ if and only if $\varphi(C) = \varphi(C)^\perp.$
\end{lemma}

\begin{proof}
    For any $x=x_1e_1 + \dots + x_te_t\in C^\perp,$ $y=y_1e_1 + \dots + y_te_t \in C,$ where $x_j\in C_j^\perp, y_j \in C_j, 1 \leq j\leq t.$ Then $\varphi(x) = (x_1,\dots,x_t),\varphi(y)=(y_1,\dots,y_t),$
    $\varphi(x)\cdot \varphi(y) = x_1\cdot y_1 + \dots + x_t\cdot y_t = 0,$ so $\varphi(C^\perp) \subset \varphi(C)^\perp .$
    Since $\varphi$ is bijective,
    \begin{equation*}
        \left|\varphi(C^\perp)\right| = \left|C^\perp\right| = \left|C_1^\perp\right|\dots \left|C_t^\perp\right|=\frac{\left|F_q^{n}\right|}{|C_1|} \dots \frac{\left|F_q^{n}\right|}{|C_t|} =\frac{q^{tn} }{|C|}= \frac{\left|F_q^{tn}\right|}{|\varphi(C)|} = \left|\varphi(C)^\perp\right|.
    \end{equation*}
    Thus, $\varphi(C^\perp) = \varphi(C)^\perp.$ The remained statement follows from this equation.
\end{proof}

\begin{theorem}
    Let $C$ be a linear code of length $n$ over $\fring.$ If $MM\trans = kE_t,$ where $k\in F_q^\ast,$ then $\eta_M \left(C^\perp\right)=\eta_M(C)^\perp,$ and $C=C^\perp$ if and only if $\eta_M (C) = \eta_M (C)^\perp.$
\end{theorem}
\begin{proof}
    For any $x=x_1e_1 + \dots + x_te_t\in C^\perp,$ $y=y_1e_1 + \dots + y_te_t \in C,$ where $x_j\in C_j^\perp, y_j \in C_j, 1 \leq j\leq t.$ Then $\varphi(x) = (x_1,\dots,x_t),\varphi(y)=(y_1,\dots,y_t),$ by Lemma \ref{lem:dualofimage} we know
    $\varphi(x)\cdot \varphi(y) = x_1\cdot y_1 + \dots + x_t\cdot y_t = 0,$ note that $\eta_M(x) = \varphi(x) (M\otimes E_n) , \eta_M(y) = \varphi(y) (M\otimes E_n),$ so
    \begin{equation*}
        \begin{aligned}
            \eta_M(x) \cdot \eta_M(y) & = \varphi(x)(M\otimes E_n) (M\otimes E_n)\trans \varphi(y)\trans \\
                                      & = k\varphi(x)\varphi(y)\trans = k\varphi(x)\cdot \varphi(y) = 0,
        \end{aligned}
    \end{equation*} thus $\eta_M(C^\perp) \subset \eta_M (C)^\perp.$
    Since $\eta_M$ is bijective,
    \begin{equation*}
        \left|\eta_M(C^\perp)\right| =\left|\varphi(C^\perp)\right| = \frac{\left|F_q^{tn}\right|}{|\varphi( C)|} = \frac{\left|F_q^{tn}\right|}{|\eta_M(C)|} = \left|\eta_M(C)^\perp\right|.
    \end{equation*}
    Thus, $\eta_M(C^\perp) = \eta_M(C)^\perp.$ From this equation the remained statement follows.
\end{proof}

Matrix product codes were first defined by Blackmore et al. \cite{blackmore2001matrix}.
\begin{definition}
    Let $C_1,C_2,\dots,C_u$ be linear codes of length $n$ over $F_q,$ $A=(a_{ij})$ be a $u\times v$ matrix over $F_q,$
    \begin{align*}
        C & =[C_1,\dots ,C_u]\cdot A \\
          & = \left\{ \left(\sum_{k=1}^u a_{k1}c_k,\sum_{k=1}^u a_{k2}c_k,\dots,\sum_{k=1}^ u a_{kv}c_k\right) \in F_q^{nv}: c_j\in C_j, 1\leq j \leq u\right\}
    \end{align*}
    is a matrix product code over $F_q.$
\end{definition}
\begin{proposition}
    If $C=C_1e_1 + \dots + C_te_t$ is a linear code of length $n$ over $\fring,$ $M$ is a $t\times t$ invertible matrix over $F_q,$ then $\eta_M( C)$ is a matrix product code over $F_q$ of length $tn,$ and its dimension is $\dim C_1 + \dots + \dim C_t.$
\end{proposition}
\begin{proof}
    Just note that $\eta_M(C) = [C_1,\dots ,C_t] \cdot M$ and $\eta_M$ is an isomorphism of vector spaces.
\end{proof}

\begin{remark}
    Since $\eta_M(C)$ is a matrix product code over $F_q,$ we establish a connection between linear codes over $\fring$ and matrix product codes over $F_q$ through $\eta_M.$ The matrix product codes that can be obtained by choosing different $M$ are very different.
\end{remark}

\begin{example} \label{ex:PlotkinSum}
    Let $C_1,C_2$ be a linear code of length $n$ over $F_q,$ we can construct a linear code of length $2n$ over $F_q.$ Specifically, $C=\{(u,u+ v)\in F_q^{2n}: u\in C_1,v\in C_2\},$ this method of constructing new code is called $(u|u+v)$ construction.
    If we set $t=2,M=\left(\begin{smallmatrix}
                1&1\\0&1
            \end{smallmatrix}\right),$ then the image of the linear code $C_1e_1 +C_2e_2$ over $F_q \times F_q$ under the isomorphism $\eta_M$ is the linear code constructed from $C_1,C_2$ by $(u|u+v)$ construction.
\end{example}

Using Magma online calculator and $(u|u+v)$ construction in Example \ref{ex:PlotkinSum}, we can construct optimal linear codes over finite fields.

\begin{example}
    In $F_3[x]$, $x^{10}-1=(x+1)(x+2)(x^4 + x^3 + x^2 + x + 1)(x^4 + 2x^3 + x^2 + 2x + 1) ,\ x^{10}+1 =(x^2+1)(x^4 + x^3 + 2x + 1)(x^4 + 2x^3 + x + 1).$ Let $C_1$ be the cyclic code generated by $g_1(x) = x+1,$ and $C_2$ be the negative cyclic code generated by $g_2(x) = x^4 + x^3 + 2x + 1 ,$ then the parameters of $C_1, C_2$ are $[10,9,2],[10,6,4],$ and the parameters of $\eta_M(C_1e_1 + C_2e_2)$ are $[20 ,15,4].$ Look up the table \cite{Grassl:codetables} we know that this is an optimal linear code, and one of its generator matrices is
    \setcounter{MaxMatrixCols}{30}
    \begin{equation*}
        \begin{pmatrix}
            1 & 0 & 0 & 0 & 0 & 0 & 0 & 0 & 0 & 1 & 0 & 0 & 0 & 0 & 0 & 0 & 1 & 2 & 0 & 2 \\
            0 & 1 & 0 & 0 & 0 & 0 & 0 & 0 & 0 & 2 & 0 & 0 & 0 & 0 & 0 & 0 & 2 & 2 & 2 & 1 \\
            0 & 0 & 1 & 0 & 0 & 0 & 0 & 0 & 0 & 1 & 0 & 0 & 0 & 0 & 0 & 0 & 1 & 1 & 2 & 1 \\
            0 & 0 & 0 & 1 & 0 & 0 & 0 & 0 & 0 & 2 & 0 & 0 & 0 & 0 & 0 & 0 & 0 & 1 & 1 & 1 \\
            0 & 0 & 0 & 0 & 1 & 0 & 0 & 0 & 0 & 1 & 0 & 0 & 0 & 0 & 0 & 0 & 1 & 2 & 1 & 0 \\
            0 & 0 & 0 & 0 & 0 & 1 & 0 & 0 & 0 & 2 & 0 & 0 & 0 & 0 & 0 & 0 & 1 & 0 & 2 & 1 \\
            0 & 0 & 0 & 0 & 0 & 0 & 1 & 0 & 0 & 1 & 0 & 0 & 0 & 0 & 0 & 0 & 1 & 0 & 0 & 1 \\
            0 & 0 & 0 & 0 & 0 & 0 & 0 & 1 & 0 & 2 & 0 & 0 & 0 & 0 & 0 & 0 & 0 & 1 & 0 & 2 \\
            0 & 0 & 0 & 0 & 0 & 0 & 0 & 0 & 1 & 1 & 0 & 0 & 0 & 0 & 0 & 0 & 0 & 0 & 1 & 1 \\
            0 & 0 & 0 & 0 & 0 & 0 & 0 & 0 & 0 & 0 & 1 & 0 & 0 & 0 & 0 & 0 & 2 & 1 & 0 & 2 \\
            0 & 0 & 0 & 0 & 0 & 0 & 0 & 0 & 0 & 0 & 0 & 1 & 0 & 0 & 0 & 0 & 1 & 1 & 1 & 1 \\
            0 & 0 & 0 & 0 & 0 & 0 & 0 & 0 & 0 & 0 & 0 & 0 & 1 & 0 & 0 & 0 & 2 & 2 & 1 & 0 \\
            0 & 0 & 0 & 0 & 0 & 0 & 0 & 0 & 0 & 0 & 0 & 0 & 0 & 1 & 0 & 0 & 0 & 2 & 2 & 1 \\
            0 & 0 & 0 & 0 & 0 & 0 & 0 & 0 & 0 & 0 & 0 & 0 & 0 & 0 & 1 & 0 & 2 & 1 & 2 & 1 \\
            0 & 0 & 0 & 0 & 0 & 0 & 0 & 0 & 0 & 0 & 0 & 0 & 0 & 0 & 0 & 1 & 2 & 0 & 1 & 1
        \end{pmatrix}.
    \end{equation*}
    In fact, the optimal linear code with parameters $[20,15,4]$ over $F_3$ in the table \cite{Grassl:codetables} is also constructed by $(u|u+v)$ construction.
\end{example}

\begin{example}
    In $F_3[x]$, $x^{11}-1=(x+2)(x^5 + 2x^3 + x^2 + 2x + 2)(x^5 + x^4 + 2x ^3 + x^2 + 2) ,\ x^{11}+1 =(x+1)(x^5 + 2x^3 + 2x^2 + 2x + 1)(x^5 + 2x^4 + 2x^3 + 2x^2 + 1).$ Let $C_1$ be the cyclic code generated by $g_1(x) = x-1$, $C_2$ be the negative cyclic code generated by $g_2(x) = x^5 + 2x^3 +  2x^2 + 2x + 1 ,$ then the parameters of $C_1, C_2$ are $[11,10,2],[11,6,5],$ and the parameters of $\eta_M(C_1e_1 + C_2e_2)$ are $[22,16,4].$ Look up the table \cite{Grassl:codetables}, we know that this is an optimal linear code, and one of its generator matrices is 
    
    \begin{equation*}\label{eq:genmatof22164}
        %G_{22}=
        \begin{pmatrix}
            1 & 0 & 0 & 0 & 0 & 0 & 0 & 0 & 0 & 0 & 2 & 0 & 0 & 0 & 0 & 0 & 0 & 1 & 2 & 2 & 2 & 2 \\
            0 & 1 & 0 & 0 & 0 & 0 & 0 & 0 & 0 & 0 & 2 & 0 & 0 & 0 & 0 & 0 & 0 & 0 & 1 & 2 & 2 & 1 \\
            0 & 0 & 1 & 0 & 0 & 0 & 0 & 0 & 0 & 0 & 2 & 0 & 0 & 0 & 0 & 0 & 0 & 1 & 2 & 0 & 1 & 1 \\
            0 & 0 & 0 & 1 & 0 & 0 & 0 & 0 & 0 & 0 & 2 & 0 & 0 & 0 & 0 & 0 & 0 & 1 & 0 & 1 & 2 & 0 \\
            0 & 0 & 0 & 0 & 1 & 0 & 0 & 0 & 0 & 0 & 2 & 0 & 0 & 0 & 0 & 0 & 0 & 2 & 2 & 1 & 2 & 1 \\
            0 & 0 & 0 & 0 & 0 & 1 & 0 & 0 & 0 & 0 & 2 & 0 & 0 & 0 & 0 & 0 & 0 & 1 & 1 & 1 & 0 & 1 \\
            0 & 0 & 0 & 0 & 0 & 0 & 1 & 0 & 0 & 0 & 2 & 0 & 0 & 0 & 0 & 0 & 0 & 1 & 0 & 0 & 0 & 2 \\
            0 & 0 & 0 & 0 & 0 & 0 & 0 & 1 & 0 & 0 & 2 & 0 & 0 & 0 & 0 & 0 & 0 & 0 & 1 & 0 & 0 & 2 \\
            0 & 0 & 0 & 0 & 0 & 0 & 0 & 0 & 1 & 0 & 2 & 0 & 0 & 0 & 0 & 0 & 0 & 0 & 0 & 1 & 0 & 2 \\
            0 & 0 & 0 & 0 & 0 & 0 & 0 & 0 & 0 & 1 & 2 & 0 & 0 & 0 & 0 & 0 & 0 & 0 & 0 & 0 & 1 & 2 \\
            0 & 0 & 0 & 0 & 0 & 0 & 0 & 0 & 0 & 0 & 0 & 1 & 0 & 0 & 0 & 0 & 0 & 2 & 1 & 1 & 1 & 0 \\
            0 & 0 & 0 & 0 & 0 & 0 & 0 & 0 & 0 & 0 & 0 & 0 & 1 & 0 & 0 & 0 & 0 & 0 & 2 & 1 & 1 & 1 \\
            0 & 0 & 0 & 0 & 0 & 0 & 0 & 0 & 0 & 0 & 0 & 0 & 0 & 1 & 0 & 0 & 0 & 2 & 1 & 0 & 2 & 1 \\
            0 & 0 & 0 & 0 & 0 & 0 & 0 & 0 & 0 & 0 & 0 & 0 & 0 & 0 & 1 & 0 & 0 & 2 & 0 & 2 & 1 & 2 \\
            0 & 0 & 0 & 0 & 0 & 0 & 0 & 0 & 0 & 0 & 0 & 0 & 0 & 0 & 0 & 1 & 0 & 1 & 1 & 2 & 1 & 1 \\
            0 & 0 & 0 & 0 & 0 & 0 & 0 & 0 & 0 & 0 & 0 & 0 & 0 & 0 & 0 & 0 & 1 & 2 & 2 & 2 & 0 & 1
        \end{pmatrix}
    \end{equation*}
    %$G_{22},$ See formula\eqref{eq:genmatof22164} .
\end{example}

Zhu et al. studied constacyclic codes over $F_p + vF_p $ ($v^2=v, $ where $p$ is an odd prime number), and the main result\cite[Theorem 3.4,][]{zhu2011class} they obtained is that the $(1-2v)$-cyclic codes of length $n$ over $F_p+vF_p$ under their Gray map are cyclic codes of length $2n$ over $F_p.$ Their results can be generalized.

Let
\begin{equation} \label{eq:specialmat}
    M = \begin{pmatrix}
        m_1 & & & \\
            &m_2 && \\
            & & \ddots & \\
            & & & m_t
    \end{pmatrix} \begin{pmatrix}
        1 & \lambda_1 & \dots & \lambda_1^{t-1} \\
        1 & \lambda_2 & \dots & \lambda_2^{t-1} \\
        \vdots & \vdots & & \vdots \\
        1 & \lambda_t & \dots & \lambda_t^{t-1}
    \end{pmatrix},
\end{equation}

where $m_j \in F_q^\ast, $ and $t\mid (q-1),$ while $\lambda_j$ are $t$ pairwise different roots of $x^t =1$ in $F_q,$ $1\leq j \leq t.$

\begin{theorem} \label{thm:skewtocyclic}
    Let the matrix $M$ as the Eq. \eqref{eq:specialmat}, $\lambda = \lambda_1e_1 + \dots + \lambda_te_t.$ If $C=C_1e_1 + \dots + C_te_t$ is a linear code of length $n$ over $\fring,$ then $C$ is a $\lambda^{-1}$-cyclic code over $\fring$ if and only if $\eta_M(C)$ is a cyclic code of length $tn$ over $F_q.$
\end{theorem}

\begin{proof}
    Denote $\rho$ the map from $F_q^{tn}$ to $F_q^{tn},$ which maps $(x_0,x_1,\dots,x_{tn-1})$ to $(x_{tn-1},x_0,\dots,x_{tn-2}).$ Use $\rho_{\lambda^{-1}_j}$ to represent the map from $F_q^{n}$ to $ F_q^{n},$ which maps $(y_0,y_1,\dots,y_{n-1})$ to $(\lambda^{-1}_j y_{n-1},y_0,\dots,y_{n-2}).$ According to Theorem \ref{thm:mainofcase1} we know that $C$ is a $\lambda^{-1}$-cyclic code over $\fring$ if and only if $C_j$ is a $\lambda^{-1}_j$-cyclic code over $F_q,$ $1\leq j\leq t.$

    Necessity. For any $x=\sum_{j=1}^t x_je_j \in C,$ we have 
    $
        \eta_M(x) = \eta_M(x_1e_1) + \dots + \eta_M(x_te_t).
    $
    Therefore, to prove that $\eta_M (C)$ is a cyclic code, it is only necessary to prove that for all $x_j =(x_{0j},x_{1j},\dots,x_{n-1,j}) \in C_j,$ $\eta_M(x_je_j) \in F_q^{tn}$ is still in $\eta_M(C)$ after a cyclic shift $\rho$. Direct calculation shows
    \begin{align*}
        \eta_M(x_je_j) & = m_j\left(x_j,\lambda_jx_j,\dots,\lambda_j^{t-1}x_j \right), \\
        \rho(\eta_M(x_je_j)) & = m_j\left( \rho_{\lambda_j^{-1}}(x_j),\lambda_j \rho_{\lambda_j^{-1}} (x_j),\dots, \lambda_j^{t-1}\rho_{\lambda_j^{-1}} (x_j) \right) \\
                             & = \eta_M\left(\rho_{\lambda_j^{-1}} (x_j)e_j\right) \in \eta_M(C).
    \end{align*}

    Sufficiency. For any $x_j \in C_j,$ we need to prove that $\rho_{\lambda_j^{-1}}(x_j) \in C_j.$ Since
    \begin{equation*}
        \rho(\eta_M(x_je_j)) = \eta_M\left(\rho_{\lambda_j^{-1}} (x_j)e_j\right) \in \eta_M(C),
    \end{equation*}
    then $\rho_{\lambda_j^{-1}} (x_j)e_j \in C,$ and $\rho_{\lambda_j^{-1}} (x_j) \in C_j.$
\end{proof}

The results of Zhu et al. can be stated in our language here:

\begin{example}\cite[Theorem 3.4,][]{zhu2011class}
    Let $C=C_1e_1 + C_2e_2$ be a linear code over $\fring= F_p\times F_p,$ where $p$ is an odd prime, $M = \left(\begin{smallmatrix}
                1&1\\
                -1&1
            \end{smallmatrix}\right),$ then $C$ is a $(e_1 - e_2)$-cyclic code over $\fring$ if and only if $\eta_M(C)$ is a cyclic code of length $2n$ over $F_p.$
\end{example}

\begin{remark}
    In Theorem \ref{thm:skewtocyclic}, we require $\lambda_j$ to be $t$ distinct roots of $x^t =1$ in $F_q$ to make $\eta_M$ be bijective, but in the proof of necessity we do not use the fact that $\eta_M$ is bijective. That is, we can define other mappings that map skew constacyclic codes over $\fring$ to cyclic codes over $F_q.$ For example, Bag et al. defined a surjective homomorphism that maps skew constacyclic codes over $F_p+uF_p + vF_p + wF_p$ to cyclic codes over $F_p,$ see \citep[Theorem 4.2][]{bag2019skew}.
\end{remark}

\subsection{Surjective homomorphism}

Consider the surjective homomorphism $\Psi: \fring \rightarrow F_q,\ \sum_{j=1}^t a_je_j
    \mapsto \sum_{j=1}^t a_j.$ A mapping from $\fring^{n}$ to $F_q^{n}$ is induced by $\Psi$, and we still denote it by $\Psi,$
i.e., $\Psi: $
\begin{align*}
    \fring^{n} & \rightarrow F_q^{n}, \\
    \left(\sum_{j=1}^t x_{0j}e_j,\dots,\sum_{j=1}^tx_{n-1,j}e_j\right) & \mapsto \left(\sum_{ j=1}^t x_{0j},\dots,\sum_{j=1}^tx_{n-1,j}\right).
\end{align*}

\begin{proposition}
    If $C=C_1e_1 + \dots + C_te_t$ is a linear code of length $n$ over $\fring,$ then $\Psi(C) = C_1+ \dots + C_t$ is a linear code of length $n$ over $F_q.$ Moreover, if the parameters of $C_i$ are $[n,k_i,d_i],$ then the parameters of $\Psi(C)$ are $[n,\dim\left( C_1+ \dots + C_t \right),d],$ where $d\leq \min\{d_1,\dots,d_t\}.$ \qed
\end{proposition}

Let $A=(1,1,\dots,1)\trans,$ then $\Psi(C) = [C_1,C_2,\dots,C_t]\cdot A,$ that is to say, $\Psi(C )$ is also a matrix product code over $F_q.$ Using the mapping $\Psi,$ we may get linear codes with good parameters. After the following proposition, we will give the specific examples.

\begin{proposition}
    Suppose that the characteristic of $F_q$ is an odd prime. If $C_1$ is a cyclic code with parameters $[n,k_1,d_1]$ over $F_q,$ $C_2$ is a negative cyclic code with parameters $[n,k_2,d_2]$ over $F_q,$ and $k_1 + k_2 < n,$ then $C_1\cap C_2 =\{0\},$ thus, the parameters of $\Psi\left(C_1e_1 +C_2e_2\right) =  C_1+C_2$ are $[n,k_1+k_2,d],$ where $d\leq \min\{d_1,d_2\}.$
\end{proposition}

\begin{proof}
    Let $a\in C_1\cap C_2,$ and $a(x)$ be the polynomial corresponding to $a$ with degree less than $n,$ then the generator polynomial $g_1(x)$ of $C_1$ divides $ a( x),$ and $C_2$'s generator polynomial $g_2(x) \mid a(x).$ Because $g_1(x) \mid x^n -1,$ $g_2(x) \mid x^n + 1,$ then $g_1(x),g_2(x)$ are coprime, so $g_1(x)g_2(x) \mid a(x),$ consider the degree of polynomials, we get $a(x) =0.$
\end{proof}

The $C_1e_1+C_2e_2$ in the above proposition is a $(e_1 -e_2)$-cyclic code over $\fring=F_q \times F_q,$ and $C_1+C_2$ is the image of $\Psi.$ Using Magma online calculator, we get some optimal linear codes over $F_3.$

\begin{example}
    In $F_3[x],$ $x^8-1=(x+1)(x+2)(x^2 + 1)(x^2 + x + 2)(x^2 + 2x + 2 ),\ x^8+1 =(x^4 + x^2 + 2)(x^4 + 2x^2 + 2).$ Let $C_1$ is the cyclic code generated by $g_1(x) = (x+1)(x +2)(x^2 + 1)(x^2 + 2x + 2),$ $C_2$ is the negative cyclic code generated by $g_2(x) = x^4 + 2x^2 + 2,$ then the parameters of $C_1$ and $C_2$ are $[8,2,6],[8,4,3]$ respectively, and the parameters of $C_1 + C_2$ are $[8,6,2].$ Look up the table \cite{Grassl:codetables}, we know that this is an optimal linear code, and one of its generator matrices is
    \begin{equation*}
        \begin{pmatrix}
            1 & 0 & 0 & 0 & 0 & 0 & 2 & 2 \\
            0 & 1 & 0 & 0 & 0 & 0 & 2 & 1 \\
            0 & 0 & 1 & 0 & 0 & 0 & 0 & 2 \\
            0 & 0 & 0 & 1 & 0 & 0 & 2 & 2 \\
            0 & 0 & 0 & 0 & 1 & 0 & 2 & 1 \\
            0 & 0 & 0 & 0 & 0 & 1 & 1 & 0
        \end{pmatrix}.
    \end{equation*}
\end{example}

\begin{example}
    In $F_3[x]$, $x^{10}-1=(x+1)(x+2)(x^4 + x^3 + x^2 + x + 1)(x^4 + 2x^3 + x^2 + 2x + 1),\ x^{10}+1 =(x^2+1)(x^4 + x^3 + 2x + 1)(x^4 + 2x^3 + x + 1).$ Let $C_1$ is the cyclic code generated by $g_1(x) = (x+1)(x^4 + x^3 + x^2 + x + 1)(x^4 + 2x^3 + x^ 2 + 2x + 1),$ $C_2$ is the negative cyclic code generated by $g_2(x) = x^2+1,$ then the parameters of $C_1$ and $C_2$ are $[10,1, 10],[10,8,2]$ respectively. And the parameters of $C_1 + C_2$ are $[10,9,2],$ Look up the table \cite{Grassl:codetables}, we know that this is an optimal linear code, one of its generator matrices is
    \begin{equation*}
        \begin{pmatrix}
            1 & 0 & 0 & 0 & 0 & 0 & 0 & 0 & 0 & 1 \\
            0 & 1 & 0 & 0 & 0 & 0 & 0 & 0 & 0 & 2 \\
            0 & 0 & 1 & 0 & 0 & 0 & 0 & 0 & 0 & 2 \\
            0 & 0 & 0 & 1 & 0 & 0 & 0 & 0 & 0 & 1 \\
            0 & 0 & 0 & 0 & 1 & 0 & 0 & 0 & 0 & 1 \\
            0 & 0 & 0 & 0 & 0 & 1 & 0 & 0 & 0 & 2 \\
            0 & 0 & 0 & 0 & 0 & 0 & 1 & 0 & 0 & 2 \\
            0 & 0 & 0 & 0 & 0 & 0 & 0 & 1 & 0 & 1 \\
            0 & 0 & 0 & 0 & 0 & 0 & 0 & 0 & 1 & 1 \\
        \end{pmatrix}.
    \end{equation*}
\end{example}

\begin{example}
    In $F_3[x]$, $x^{12}-1=(x+1)^3(x+2)^3(x^2 + 1)^3,\ x^{12}+1 = (x^2 +x+ 2)^3(x^2 + 2x + 2)^3.$ Let $C_1$ is the cyclic code generated by $g_1(x) = (x+1)(x+2)(x^2 + 1 )^3,$ $C_2$ is the negative cyclic code generated by $g_2(x) = (x^2 +x+ 2)(x^2 + 2x + 2)^3,$ then the parameters of $C_1,C_2 $ are $[12,4,4],[12,4,6],$ and the parameters of $C_1 + C_2$ are $[12,8,3].$ Look up table \cite{Grassl:codetables}, we know that this is an optimal linear code, and one of its generator matrices is
    \begin{equation*}
        \begin{pmatrix}
            1 & 0 & 0 & 0 & 0 & 0 & 0 & 0 & 2 & 1 & 2 & 1 \\
            0 & 1 & 0 & 0 & 0 & 0 & 0 & 0 & 1 & 1 & 2 & 1 \\
            0 & 0 & 1 & 0 & 0 & 0 & 0 & 0 & 0 & 1 & 1 & 2 \\
            0 & 0 & 0 & 1 & 0 & 0 & 0 & 0 & 0 & 0 & 1 & 1 \\
            0 & 0 & 0 & 0 & 1 & 0 & 0 & 0 & 2 & 1 & 2 & 2 \\
            0 & 0 & 0 & 0 & 0 & 1 & 0 & 0 & 0 & 2 & 1 & 2 \\
            0 & 0 & 0 & 0 & 0 & 0 & 1 & 0 & 0 & 0 & 2 & 1 \\
            0 & 0 & 0 & 0 & 0 & 0 & 0 & 1 & 2 & 1 & 2 & 0 \\
        \end{pmatrix}.
    \end{equation*}
\end{example}

\section{Conclusions}
In this article, we study skew constacyclic codes over a class of finite commutative semisimple rings. The automorphism group of $\fring$ is determined, and we characterize skew constacyclic codes over ring by linear codes over finite field. We also define homomorphisms which map linear codes over $\fring$ to matrix product codes over $F_q,$ some optimal linear codes over finite fields are obtained.

\section{Acknowledgements}
This article contains the main results of the author's thesis of master degree. This research did not receive any specific grant from funding agencies in the public, commercial, or
not-for-profit sectors.

\bibliography{mybibfile}

\end{document}